\documentclass[a4paper,UKenglish,cleveref, autoref, thm-restate]{lipics-v2021}
\usepackage{mathtools,stmaryrd}
\usepackage{vwcol} 
\usepackage{multicol}
\usepackage[normalem]{ulem}
\usepackage{subcaption}
\usepackage{enumitem}

\newcommand{\lovbf}{\ensuremath{\textbf{\textup{LO}}_\textbf{\textup{v}}}}
\newcommand{\pplovbf}{\ensuremath{\textbf{\textup{LO}}_\textbf{\textup{PP}}}}
\newcommand{\lov}{\ensuremath{\textup{LO}_\textup{v}}}
\newcommand{\lopp}{\ensuremath{\textup{LO}_\textup{PP}}}

\title{\texorpdfstring{\lov}{LOv}-Calculus: A Graphical Language for Linear Optical Quantum Circuits}
\author{~}{~}{}{}{}

\author{Alexandre Cl\'ement}{Universit\'e de Lorraine, CNRS, Inria, LORIA, F-54000 Nancy, France \and \url{https://members.loria.fr/AClement} }{alexandre.clement@loria.fr}{https://orcid.org/0000-0002-7958-5712}{}
\author{Nicolas Heurtel}{Quandela, 7 Rue Léonard de Vinci, 91300 Massy, France}{nicolas.heurtel@quandela.com}{}{}

\author{Shane Mansfield}{Quandela, 7 Rue Léonard de Vinci, 91300 Massy, France}{shane.mansfield@quandela.com}{}{}
\author{Simon Perdrix}{Universit\'e de Lorraine, CNRS, Inria, LORIA, F-54000 Nancy, France \and \url{https://members.loria.fr/SPerdrix} }{simon.perdrix@loria.fr}{https://orcid.org/0000-0002-1808-2409}{}

\author{Benoît Valiron}{Université Paris-Saclay, Inria, CNRS, ENS
  Paris-Saclay, CentraleSupélec, LMF, 91190,
Gif-sur-Yvette, France.}{benoit.valiron@universite-paris-saclay.fr}{}{}
\authorrunning{Cl\'ement, Heurtel, Mansfield, Perdrix and Valiron} 



\category{} 

\relatedversion{} 

\supplement{}

\funding{This work is funded by ANR-17-CE25-0009 SoftQPro, ANR-17-CE24-0035 VanQuTe,
PIA-GDN/Quantex, and LUE / UOQ, PEPR EPIQ, and by  \emph{``Investissements d'avenir''} (ANR-15-IDEX-02) program of
the French National Research Agency. 
}

\nolinenumbers
\hideLIPIcs

\EventEditors{John Q. Open and Joan R. Access}
\EventNoEds{2}
\EventLongTitle{ICALP}
\EventShortTitle{ICALP 2020}
\EventAcronym{ICALP}
\EventYear{2016}
\EventDate{December 24--27, 2016}
\EventLocation{Little Whinging, United Kingdom}
\EventLogo{}
\SeriesVolume{42}
\ArticleNo{1}

\usepackage[svgnames,dvipsnames]{xcolor}
\usepackage{tikzit}

\tikzstyle{diamant}=[diamond, fill=couleurdefond, draw=black]
\tikzstyle{newe}=[rectangle, fill={gray!15}, draw=black, tikzit shape=rectangle, inner sep=0.2em]
\tikzstyle{cercle}=[circle, fill=couleurdefond, draw=black]
\tikzstyle{scercle}=[circle, fill=couleurdefond, draw=black, tikzit fill=white, inner sep=0.1em]
\tikzstyle{cartouche}=[rounded rectangle, fill=couleurdefond, draw=black]
\tikzstyle{neg}=[rounded rectangle, fill=couleurdefond, draw=black, execute at end node={$\neg$}]
\tikzstyle{sneg}=[rounded rectangle, fill=couleurdefond, draw=black, execute at end node={$\neg$}, scale=0.8]
\tikzstyle{negserie}=[rounded rectangle, fill=couleurdefond, draw=black, execute at end node={\footnotesize$\star\star$}]
\tikzstyle{diagrammevide}=[rectangle, fill=couleurdefond, draw=black, inner sep=1.25em, borddiagrammevide, tikzit shape=rectangle]
\tikzstyle{mdiagrammevide}=[rectangle, fill=couleurdefond, draw=black, inner sep=0.75em, sborddiagrammevide, tikzit shape=rectangle]
\tikzstyle{msdiagrammevide}=[rectangle, fill=couleurdefond, draw=black, inner sep=0.7em, msborddiagrammevide, tikzit shape=rectangle]
\tikzstyle{sdiagrammevide}=[rectangle, fill=couleurdefond, draw=black, inner sep=0.5em, sborddiagrammevide, tikzit shape=rectangle]
\tikzstyle{xsdiagrammevide}=[rectangle, fill=couleurdefond, draw=black, inner sep=0.4em, xsborddiagrammevide, tikzit shape=rectangle]
\tikzstyle{bs}=[shape=beam, fill=couleurdefond, draw, inner sep=0.25em, thick, tikzit fill=white]
\tikzstyle{sbs}=[shape=beam, fill=couleurdefond, draw, inner sep=0.2em, thick, tikzit fill=white]
\tikzstyle{npbs}=[shape=beam, horizontal fill={{npbsmoitiebasse}{npbsmoitiehaute}}, draw, inner sep=0.25em, thick, tikzit fill={rgb,255: red,128; green,128; blue,128}]
\tikzstyle{npbsalenvers}=[shape=beam, horizontal fill={{npbsmoitiehaute}{npbsmoitiebasse}}, draw, inner sep=0.25em, thick, tikzit fill={rgb,255: red,128; green,128; blue,128}]
\tikzstyle{snpbs}=[shape=beam, horizontal fill={{npbsmoitiebasse}{npbsmoitiehaute}}, draw, inner sep=0.2em, thick, tikzit fill={rgb,255: red,128; green,128; blue,128}]
\tikzstyle{snpbsalenvers}=[shape=beam, horizontal fill={{npbsmoitiehaute}{npbsmoitiebasse}}, draw, inner sep=0.2em, thick, tikzit fill={rgb,255: red,128; green,128; blue,128}]
\tikzstyle{cnot}=[shape=circle, draw, path picture={ 
\draw[black](path picture bounding box.north) -- (path picture bounding box.south) (path picture bounding box.west) -- (path picture bounding box.east);
}]
\tikzstyle{thickcnot}=[shape=circle, draw, thick, path picture={ 
\draw[thick,black](path picture bounding box.north) -- (path picture bounding box.south) (path picture bounding box.west) -- (path picture bounding box.east);
}]
\tikzstyle{boite22}=[fill=white, draw=black, shape=rectangle, minimum height=1cm, minimum width=0.5cm]
\tikzstyle{boite15}=[fill=white, draw=black, shape=rectangle, minimum height=0.7cm, minimum width=0.5cm]
\tikzstyle{boite2}=[fill=white, draw=black, shape=rectangle, minimum height=0cm, minimum width=0cm]
\tikzstyle{snegpotentiel}=[fill=couleurdefond, draw=black, shape=rounded rectangle, inner sep=0.25em, tikzit fill={rgb,255: red,191; green,191; blue,191}, execute at end node={\footnotesize$\star$}]
\tikzstyle{negpotentiel}=[fill=couleurdefond, draw=black, shape=rounded rectangle, tikzit fill={rgb,255: red,191; green,191; blue,191}, execute at end node={$\star$}]
\tikzstyle{token}=[fill=black, draw=black, shape=circle, inner sep=0.1em]
\tikzstyle{whitetoken}=[fill=white, draw=black, shape=circle, inner sep=0.1em]
\tikzstyle{boitePBS}=[fill=white, draw=gray, thick, shape=rectangle, rounded corners=3pt, minimum height=0.6cm, inner sep=0.1em, minimum width=0.5cm]
\tikzstyle{boitePBS2}=[fill=white, draw=gray, thick, shape=rectangle, rounded corners=3pt, minimum height=0.55cm, inner sep=0.1em, minimum width=0.5cm]
\tikzstyle{sgene}=[fill={gray!30}, draw=black, shape=rounded rectangle, rounded rectangle east arc=0pt, minimum height=0.5cm, inner sep=0em, minimum width=0cm, scale=0.8]
\tikzstyle{sdetector}=[fill={gray!30}, draw=black, shape=rounded rectangle, rounded rectangle west arc=0pt, minimum height=0.5cm, inner sep=0em, minimum width=0cm, scale=0.8]
\tikzstyle{xsgene}=[fill={gray!30}, draw=black, shape=rounded rectangle, rounded rectangle east arc=0pt, minimum height=0.5cm, inner sep=0em, minimum width=0cm, scale=0.67]
\tikzstyle{xsdetector}=[fill={gray!30}, draw=black, shape=rounded rectangle, rounded rectangle west arc=0pt, minimum height=0.5cm, inner sep=0em, minimum width=0cm, scale=0.67]
\tikzstyle{PolRot}=[fill={gray!30}, draw=black, shape=rectangle, minimum height=0.5cm, inner sep=0.1em, minimum width=0.1cm]
\tikzstyle{PolRotrouge}=[fill={red!3}, draw=red, shape=rectangle, minimum height=0.5cm, inner sep=0.1em, minimum width=0.1cm, tikzit fill=red, execute at begin node={\textcolor{red}\bgroup}, execute at end node={\egroup}]
\tikzstyle{PhS}=[fill=white, draw=black, shape=rectangle, minimum height=0.5cm, inner sep=0.1em, minimum width=0.1cm]
\tikzstyle{gene}=[fill={gray!30}, draw=black, shape=rounded rectangle, rounded rectangle east arc=0pt, minimum height=0.5cm, inner sep=0em, minimum width=0cm]
\tikzstyle{detector}=[fill={gray!30}, draw=black, shape=rounded rectangle, rounded rectangle west arc=0pt, minimum height=0.5cm, inner sep=0em, minimum width=0cm]
\tikzstyle{cartoucherouge}=[rounded rectangle, fill={red!55!white}, draw=black, tikzit fill=red]
\tikzstyle{cartouchebleu}=[rounded rectangle, fill={blue!33!white}, draw=black, tikzit fill=blue]
\tikzstyle{diamantrouge}=[diamond, fill={rgb,255: red,255; green,115; blue,115}, draw=black]
\tikzstyle{diamantbleu}=[diamond, fill={rgb,255: red,171; green,171; blue,255}, draw=black]

\tikzstyle{new}=[-]
\tikzstyle{tirets}=[-, draw=black, dashed]
\tikzstyle{noire}=[-, draw=black]
\tikzstyle{ep}=[-, draw=black]
\tikzstyle{longdashed}=[-, dash pattern=on 5pt off 5pt]
\tikzstyle{pointilles}=[-, draw=black, dotted]
\tikzstyle{grise}=[-, draw={rgb,255: red,191; green,191; blue,191}]
\tikzstyle{rouge}=[-, draw=red]
\tikzstyle{bleue}=[-, draw=bleu, tikzit draw=blue]
\tikzstyle{verte}=[-, draw={rgb,255: red,0; green,230; blue,0}]
\tikzstyle{borddiagrammevide}=[-, dash pattern=on 0.5em off 0.5em on 0.5em off 0.5em on 0.5em off 0em]
\tikzstyle{msborddiagrammevide}=[-, dash pattern=on 0.28em off 0.28em on 0.28em off 0.28em on 0.28em off 0em]
\tikzstyle{sborddiagrammevide}=[-, dash pattern=on 0.2em off 0.2em on 0.2em off 0.2em on 0.2em off 0em]
\tikzstyle{xsborddiagrammevide}=[-, dash pattern=on 0.1em off 0.1em on 0.15em off 0.1em on 0.1em off 0em]
\tikzstyle{mediumdash}=[-, dash pattern=on 2pt off 2pt]
\tikzstyle{rougefonce}=[-, draw={red!50!black}, tikzit draw={rgb,255: red,136; green,0; blue,0}]

\input{figures/styles-pbs.tikzdefs}

\renewcommand{\H}{\textup H}
\newcommand{\V}{\textup V}

\hypersetup{hypertexnames=false,raiselinks=true}
\usepackage{stmaryrd}
\usepackage{scalerel}
\usepackage{longtable}

\newcolumntype{C}{>{$}c<{$}}  
\newcolumntype{R}{>{$}r<{$}}  
\newcolumntype{L}{>{$}l<{$}}  
\DeclareFontShape{OMX}{cmex}{m}{n}{
  <-7.5> cmex7
  <7.5-8.5> cmex8
  <8.5-9.5> cmex9
  <9.5-> cmex10
}{}
\SetSymbolFont{largesymbols}{normal}{OMX}{cmex}{m}{n}
\SetSymbolFont{largesymbols}{bold}  {OMX}{cmex}{m}{n}

\newcommand{\interp}[1]{\left\llbracket #1 \right\rrbracket}
\newcommand{\interps}[1]{\left\llbracket #1 \right\rrbracket}
\newcommand\interph[1]{\interp{#1}_{\mathrm{pp}}}

 {\everymath{\displaystyle\everymath{}}\array}%
 {\endarray}
 {\everymath{\scriptstyle\everymath{}}\array}%
 {\endarray}

\newcommand\Hi{\mathcal H}

\newcommand\R{\mathbb R}

\newcommand\CC{\mathbb C}

\newcommand\ii{i}

\newcommand{\scalprod}[2]{\langle #1|#2\rangle}

\newlength{\xlutmvcyp}

\makeatletter
\newlength{\negph@wd}
\DeclareRobustCommand{\negphantom}[1]{%
  \ifmmode
    \mathpalette\negph@math{#1}%
  \else
    \negph@do{#1}%
  \fi
}
\newcommand{\negph@math}[2]{\negph@do{$\m@th#1#2$}}
\newcommand{\negph@do}[1]{%
  \settowidth{\negph@wd}{#1}%
  \hspace*{-\negph@wd}%
}
\makeatother

\makeatletter
\newlength{\halfnegph@wd}
\DeclareRobustCommand{\halfnegphantom}[1]{%
  \ifmmode
    \mathpalette\halfnegph@math{#1}%
  \else
    \halfnegph@do{#1}%
  \fi
}
\newcommand{\halfnegph@math}[2]{\halfnegph@do{$\m@th#1#2$}}
\newcommand{\halfnegph@do}[1]{%
  \settowidth{\halfnegph@wd}{#1}%
  \hspace*{-0.5\halfnegph@wd}%
}
\makeatother

\makeatletter
\newlength{\halfph@wd}
\DeclareRobustCommand{\halfphantom}[1]{%
  \ifmmode
    \mathpalette\halfph@math{#1}%
  \else
    \halfph@do{#1}%
  \fi
}
\newcommand{\halfph@math}[2]{\halfph@do{$\m@th#1#2$}}
\newcommand{\halfph@do}[1]{%
  \settowidth{\halfph@wd}{#1}%
  \hspace*{0.5\halfph@wd}%
}
\makeatother

\newcommand{\lowindex}
{}

\newcommand\eqeqref[1]{\overset{\eqref{#1}}{=}}

\newcommand\eqdeuxeqref[2]{\overset{\eqref{#1}\eqref{#2}}{=}}
\newcommand\eqtroiseqref[3]{\overset{\eqref{#1}\eqref{#2}\eqref{#3}}{=}}

\newlength\traitsdiagrammevide
\newcounter{eqnabc}

\newcounter{eqnABC}
\newenvironment{eqnABC}{\refstepcounter{eqnABC}\equation}{\tag{\Alph{eqnABC}}\endequation}

\newcommand{\Scom}[1]{\ifdefined\commentaire{\begin{color}{green!30!black} \texttt{\small[SP: #1]}\end{color}}\fi}%

\newcommand{\alexandrecom}[1]{\ifdefined\commentaire{\begin{color}{orange!80!black} \texttt{\small[A: #1]}\end{color}}\fi}

\newcommand{\urlalt}[2]{\href{#2}{\nolinkurl{#1}}}

\theoremstyle{definition}

\renewcommand{\ket}[1]{|{#1}\rangle}

\begin{document}

\maketitle

\begin{abstract}
  We introduce the $\lov$-calculus, a graphical language for reasoning about linear optical quantum circuits with so-called vacuum state auxiliary inputs. 
  We present the axiomatics of the language and prove its soundness and completeness: two $\lov$-circuits represent the same quantum process if and only if one can be transformed into the other with the rules of the LO$_\textup{v}$-calculus.
  We give a confluent and terminating rewrite system to rewrite any polarisation-preserving LO$_\textup{v}$-circuit into a unique triangular normal form, inspired by the universal decomposition of Reck \textit{et al.}\ (1994) for linear optical quantum circuits. 
\end{abstract}

\section{Introduction}

\newcommand{\typ}[3][\mathcal H]{{#1}^{(#3)}}

Quantum computing and information processing promise a variety of advantages
over their classical analogues, from the potential for computational speedups
(e.g.\ \cite{grover,shor}) to enhanced security and communication (e.g.\
\cite{BENNETT20147,ekert}).
By encoding information into the states of physical systems that are quantum
rather than classical,
one can then process that information by evolving and manipulating the systems
according to the laws of quantum mechanics.
This opens up the possibility of exploiting non-classical behaviours
available to quantum systems in order to process information in radically new
and potentially advantageous ways.

The development of quantum technologies has proceeded at pace over the past
number of years,
with a variety of different physical supports for quantum information being
pursued.
These include matter-based systems like superconducting circuits, cold atoms,
and trapped ions,
as well as light-based systems, in which information is encoded in photons.
Among these, photons have a privileged role in the sense that regardless of
hardware choice it will eventually be necessary to network quantum processors,
and (as the only sensible support for communicating quantum information)
some quantum information will need to be treated photonically.
Yet, in their own right, photons also offer viable approaches to quantum
computation in the noisy intermediate-scale \cite{klm} and large-scale
fault-tolerant \cite{bartolucci2021fusion} regimes.

The standard unit of quantum information is the quantum bit or qubit, and
photons allow for a rich variety of ways to encode qubits.
However it is also interesting to note that treating photons as informational
units in their own right can be advantageous.
A good example is BosonSampling, originally proposed by Aaronson and Arkhipov
\cite{bosonsampling},
a computational task that is $\# P$-hard but which can be efficiently solved by
interacting photons in an idealised generic linear-optical circuit
in which no qubit encoding need be imposed. 
At present, along with Random Circuit Sampling
\cite{aaronson2016complexitytheoretic,Bouland2019},
this provides one of the two main routes to experimental demonstrations of
quantum computational advantage
\cite{arute2019quantum,zhong2020quantum,wu2021strong,zhong2021phase},
in which quantum devices have been claimed to outperform classical supercomputing capabilities for
specific tasks.

The usual semantics for quantum computation stemming from quantum
mechanics is based on unitary matrices (or unitary operators in
general) over Hilbert spaces. Although this faithfully models the
extensional behaviour of a computation, it fails to address several key
aspects that are of interest when considering the design and
implementation of quantum algorithms. A first limitation is the
intensional description of the computation: an algorithm or quantum
computation in general consists of modular components that are
composed and combined in specific way, and one wants to keep track of
this information. One therefore needs a \emph{language} for coding
these. The other important aspect is the need to specify and verify
the said code. Indeed, classically simulating a quantum process is a
task that is exponentially costly in the size of the system, while running
code on physical devices is expensive. If some limited testing
techniques are available on quantum
systems~\cite{feng2015qpmc,li2020projection-based}, it is however
highly desirable to be able to reason and prove the desired properties
of the code upstream, and rely on \emph{formal methods}.
If text-based high-level languages oriented towards formal methods
have successfully been proposed in the
literature~\cite{green2013quipper,silq,javadiabhari2015scaffcc}, we
aim in this paper to explore a lower-level, graphical language, making
contact with photonic hardware.

Graphical languages for quantum computation have a long history:
since Feyman diagrams~\cite{feynman1965quantum}, graphical
languages for representing (low-level) quantum processes have been considered
as an answer to the limitations of plain unitary matrices.  Quantum
circuits -- the quantum equivalent to classical, boolean circuits --
are an obvious candidate for a graphical language, and indeed, several
lines of research took them as their main object of
study~\cite{green2013quipper,goisync,paykin2017qwire,chareton-qbricks}.
Quantum circuits in particular form a natural medium for describing
the execution flow of a computation.
The main problem with the model of quantum circuits is the lack of
a satisfactory equational presentation. If several attempts have been
made for various
subsets~\cite{cockett2018category,cockett2018categorytof,hutslar2018library,makary2021generators},
none of them provides a complete presentation.

A recent proposal responding to the shortfalls of quantum circuits as a model
is the ZX-calculus \cite{coecke2017picturing}, which, along with its variants
\cite{carette2019szx,backens2019zh,carette2019completeness}, have proved to
be particularly useful for reasoning about qubit quantum mechanics,
for applications such as
quantum circuit optimisation \cite{duncan2020graph,backens2021there},
verification
\cite{duncan2013verifying,garvie2018verifying,hillebrand2011quantum}
and representation e.g.\ for MBQC patterns \cite{duncan2010rewriting} or
error-correction \cite{duncan2010rewriting,beaudrap2017zx}.
However, while ZX-calculus is versatile and provides a welcomed formal
semantics for quantum computation, it remains at an abstract level.

There is therefore a clear interest in developing a graphical language
for quantum photonic processes, especially linear quantum optics,
which is closer to photonic hardware and laboratory operations that
are easily implementable in bulk optics, fibres, or in integrated
photonic circuits.
This would provide a more formal counterpart to software frameworks that have
been proposed for defining and classically simulating such processes to the
extent that it is tractable
\cite{killoran2019strawberry,heurtel2022perceval}.
The need for such a formal language is also evidenced, for example,
by the appeal to diagrams
to concisely illustrate equivalent unitaries in recent work in the Physics
literature \cite{pont2022quantifying}.
Following on the trend for graphical quantum
languages, the $\text{PBS-calculus}$~\cite{alex2020pbscalculus,branciard2021coherent,clement2022minimising} has been
proposed as a first step towards an alternative to ZX dedicated to
linear quantum optical computation (LOQC). The PBS-calculus makes it
possible to reason on a small subset of linear optical components only
acting on the polarisation of a photon. While it is enough to describe
and analyse non causally-ordered computations, it fails short at
expressing other aspects of LOQC typically considered in the Physics
community, such as the phase.

Our goal here is to take a more bottom-up approach and to propose a new
language which formalises the kinds of diagrammatics that are currently in
use in the Physics community.
In practice this can find many uses including for the design, optimisation,
verification, error-correction, and
systematic study of linear optical quantum circuits for quantum information.

Our main contributions are the following.
\begin{itemize}
\item A graphical language for LOQC featuring most of the
  physical apparatuses used in the Physics literature. The language
  comes equipped with an equational theory that is sound and complete with
  respect to the standard semantics of LOQC.

\item A strongly normalising and globally confluent rewrite system and
  normal form for the polarisation-preserving fragment, for which we
  recover the Reck \textit{et al.}\ \cite{Reck1994unitary} decomposition as
  normal form (modulo $0$-angled beam splitters and $0$-angled phase
  shifters) with a novel proof of its uniqueness.

\end{itemize}
Finally, and maybe more importantly, our language makes it possible to
formalise and reason within a common framework on various
presentations of LOQC stemming from parallel research paths.  Our
semantics not only allow us to recover, extend and improve on some key
results in LOQC such as the universal decompositions of Reck
\textit{et al.}\  \cite{Reck1994unitary} and Clements \textit{et al.}\
\cite{Clements2016unitary}, but it also gives a unifying language for
the different formalisms from the literature.

The article is structured as follows.
In \cref{section2}, we present the syntax and the semantics of the
\lov-calculus. The equational theory and its soundness are given in \cref{section3}.
In \cref{section:pp} we present the strongly normalising and globally
confluent rewrite system.
This allows us to prove the completeness of the \lov-calculus in \cref{subsect:completeness}.
Finally, we conclude in \cref{conclusion}.

\begin{figure}[tb]
  \begin{subfigure}[b]{0.61\linewidth}
    \centerline{\input{triangleetoiles-8modes-star-xs.tikz}}
    \caption{Triangular form \cite{Reck1994unitary}.}\label{triangleNF}
  \end{subfigure}
  \begin{subfigure}[b]{0.39\linewidth}
    \input{rectanglewithoutphasesetoiles-nouveau-xs.tikz}
    \caption{Rectangular form \cite{Clements2016unitary}.}\label{rectangleNF}
  \end{subfigure}
  \caption{Triangular and rectangular forms for
    polarisation-preserving circuits. 
\Scom{Est-ce que c'est toujours pertinent de mettre la rectangular form ici?}
}
  \label{fig:canon}
\end{figure}

\section{Linear Optical Quantum Circuits}
\label{section2}

A linear optical quantum computation~\cite{kok2007linear,kok2010introduction} (LOQC)
consists of spatial modes through which photons pass --
which may be physically instantiated by optical fibers, waveguides in
integrated circuits, or simply by paths in free space (bulk optics) --
and operations that act on the spatial and polarisation degrees of freedom of
the photons,
including in particular \emph{beam splitters} (\input{bs-xs.tikz}),
\emph{polarising beam splitters}
(\input{beamsplitter-xs.tikz}), \emph{phase shifters}
(\begin{tikzpicture}[scale=0.3]
	\begin{pgfonlayer}{nodelayer}
		\node [style=none] (5) at (-1.6, 0) {};
		\node [style=none] (9) at (1.6, 0) {};
		\node [style=PolRot,scale=0.75] (11) at (0, 0) {\small $\varphi$};
	\end{pgfonlayer}
	\begin{pgfonlayer}{edgelayer}
		\draw (5.center) to (11);
		\draw (11) to (9.center);
	\end{pgfonlayer}
\end{tikzpicture}
), \emph{wave plates}
(\begin{tikzpicture}[scale=0.3]
	\begin{pgfonlayer}{nodelayer}
		\node [style=none] (5) at (-1.4, 0) {};
		\node [style=none] (9) at (1.4, 0) {};
		\node [style=PhS,scale=0.85] (11) at (0, 0) {};
		\node [style=none] (12) at (1, -1) {\scriptsize $\theta$};
	\end{pgfonlayer}
	\begin{pgfonlayer}{edgelayer}
		\draw (5.center) to (11);
		\draw (11) to (9.center);
	\end{pgfonlayer}
\end{tikzpicture}
), \emph{pola-negations} \mbox{(\begin{tikzpicture}[scale=0.6]
	\begin{pgfonlayer}{nodelayer}
		\node [style=cartouche,scale=0.9] (0) at (0, 0) {\tiny$\neg$};
		\node [style=none] (1) at (-1.1, 0) {};
		\node [style=none] (2) at (1.1, 0) {};
	\end{pgfonlayer}
	\begin{pgfonlayer}{edgelayer}
		\draw [style=noire] (1.center) to (0);
		\draw [style=noire] (0) to (2.center);
	\end{pgfonlayer}
\end{tikzpicture}
)} and
finally the \emph{vacuum state sources} and \emph{detectors}
(\begin{tikzpicture}[scale=0.7]
	\begin{pgfonlayer}{nodelayer}
		\node [style=none] (2) at (1.1, 0) {};
		\node [style=gene,scale=0.7] (3) at (0, 0) { $0$};
	\end{pgfonlayer}
	\begin{pgfonlayer}{edgelayer}
		\draw (3) to (2.center);
	\end{pgfonlayer}
\end{tikzpicture}
 and \begin{tikzpicture}[scale=0.7]
	\begin{pgfonlayer}{nodelayer}
		\node [style=none] (2) at (-1.1, 0) {};
		\node [style=detector,scale=0.7] (3) at (0, 0) {$0$};
	\end{pgfonlayer}
	\begin{pgfonlayer}{edgelayer}
		\draw (3) to (2.center);
	\end{pgfonlayer}
\end{tikzpicture}
).
Their action and the semantics are described in Section~\ref{sec:sem}.

\subsection{Syntax}

\def\btikzfig#1{$\scalebox{.2}{\input{#1.tikz}}$}
In order to formalise linear optical quantum circuits, we use the formalism of
PROPs~\cite{prop}. A PRO is a strict monoidal category
whose monoid of objects is freely generated by a single $X$:
the objects are all of the form $X\otimes ...\otimes X$, and simply
denoted by $n$, the number of occurrences of $X$. PROs are typically
represented graphically as circuits: each copy of $X$ is represented by a wire
and morphisms by boxes on wires,
so that $\oplus$ is represented vertically and morphism
composition ``$\circ$'' is represented horizontally.
For instance, $D_1$ and $D_2$ represented as \input{D1-3wires.tikz} and
\input{D2-3wires.tikz} can be horizontally composed as
$D_2\circ D_1$, represented by \input{D1compD2-xs.tikz}, and vertically
composed as $D_1\oplus D_2$, represented by\quad $\input{D1tensD2-xs.tikz}$~. A PROP is the symmetric  monoidal analogue of PRO, so it is equipped with a swap \input{swap-xs.tikz}.

\begin{definition}
  \label{def:prop}
  $\lovbf$ is the PROP of \lov-circuits generated by
  $$\begin{tikzpicture}
	\begin{pgfonlayer}{nodelayer}
		\node [style=none] (2) at (1.1, 0) {};
		\node [style=gene] (3) at (0, 0) {$0$};
	\end{pgfonlayer}
	\begin{pgfonlayer}{edgelayer}
		\draw (3) to (2.center);
	\end{pgfonlayer}
\end{tikzpicture}
:0 \to 1\qquad\qquad  \begin{tikzpicture}
	\begin{pgfonlayer}{nodelayer}
		\node [style=none] (2) at (-1, 0) {};
		\node [style=detector] (3) at (0, 0) {$0$};
	\end{pgfonlayer}
	\begin{pgfonlayer}{edgelayer}
		\draw (3) to (2.center);
	\end{pgfonlayer}
\end{tikzpicture}
:1 \to 0 \qquad\qquad  \begin{tikzpicture}
	\begin{pgfonlayer}{nodelayer}
		\node [style=none] (5) at (-1.1, 0) {};
		\node [style=none] (9) at (1.1, 0) {};
		\node [style=PolRot] (11) at (0, 0) {$\varphi$};
	\end{pgfonlayer}
	\begin{pgfonlayer}{edgelayer}
		\draw (5.center) to (11);
		\draw (11) to (9.center);
	\end{pgfonlayer}
\end{tikzpicture}
:1 \to 1 \qquad\qquad \begin{tikzpicture}
	\begin{pgfonlayer}{nodelayer}
		\node [style=none] (5) at (-1.1, 0) {};
		\node [style=none] (9) at (1.1, 0) {};
		\node [style=PhS] (11) at (0, 0) {};
		\node [style=none] (12) at (0.5, -0.4) {$\theta$};
	\end{pgfonlayer}
	\begin{pgfonlayer}{edgelayer}
		\draw (5.center) to (11);
		\draw (11) to (9.center);
	\end{pgfonlayer}
\end{tikzpicture}
:1 \to 1$$ 
  $$
  \input{bs.tikz}:2 \to 2 \qquad\qquad \input{beamsplitter.tikz}:2 \to 2$$
  where $\theta, \varphi\in\mathbb R$.
  When the parameters $\theta$ and $\varphi$ are omitted we take them to be
  equal to $\pi/4$.
  We write $\begin{tikzpicture}[scale=0.8]
	\begin{pgfonlayer}{nodelayer}
		\node [style=cartouche] (0) at (0, 0) {\footnotesize$\neg$};
		\node [style=none] (1) at (-0.8, 0) {};
		\node [style=none] (2) at (0.8, 0) {};
	\end{pgfonlayer}
	\begin{pgfonlayer}{edgelayer}
		\draw [style=noire] (1.center) to (0);
		\draw [style=noire] (0) to (2.center);
	\end{pgfonlayer}
\end{tikzpicture}
$ as a shortcut notation for
  $\input{pisur2pisur2.tikz}$.
  The tensor of the monoidal structure is denoted with $\oplus$, and the
  identity, swap and empty circuit (unit of $\oplus$) are denoted as
  follows: \scalebox{.9}{$\begin{tikzpicture}
	\begin{pgfonlayer}{nodelayer}
		\node [style=none] (0) at (-0.8, 0) {};
		\node [style=none] (1) at (0.8, 0) {};
	\end{pgfonlayer}
	\begin{pgfonlayer}{edgelayer}
		\draw [style=noire] (0.center) to (1.center);
	\end{pgfonlayer}
\end{tikzpicture}:1 \to 1,\quad\input{swap-s.tikz}:2
  \to 2,\quad\begin{tikzpicture}[scale=0.3]
	\begin{pgfonlayer}{nodelayer}
		\node [style=sdiagrammevide] (0) at (0, 0) {};
	\end{pgfonlayer}
\end{tikzpicture}
:0\to 0$.}
\end{definition}

\begin{figure}[t]
  \scalebox{.8}{\input{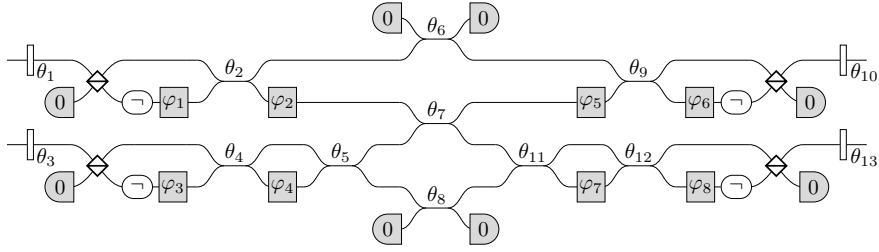}}
  \caption{\lov-circuit implementing a variational quantum eigensolver
  \cite{peruzzo2014variational}, an algorithm with applications including
  calculation of ground-state energies in quantum chemistry.}\label{vqecircuit}
\end{figure}
  
\begin{figure}[t]
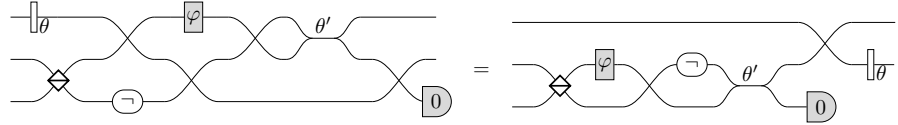

  \begin{multicols}{2}
    \scalebox{.8}{$\begin{array}{rcl} \input{ex_deformation1.tikz} &=& \input{ex_deformation2.tikz}   \end{array}$}
  \end{multicols}
  \caption{Two equivalent representations of the same \lov-circuit.}\label{circuitdeformation}
\end{figure}

\begin{example}
  An example of a linear optical quantum circuit using all of the connectives
  presented in Definition~\ref{def:prop} is shown in Figure~\ref{vqecircuit}.
\end{example}

\begin{remark}
  The axioms of PROPs guarantee that linear optical quantum circuits are
  defined up to deformations: Figure~\ref{circuitdeformation} shows two
  equivalent circuits under the equations of PROPs.
\end{remark}

Among the generators, the {beam splitters} and {phase shifters} are known to preserve the polarisation of the photons, as a consequence, we define a \emph{polarisation-preserving} sub-PRO of $\lov$ as follows. 
\begin{definition}
$\pplovbf$ is the PRO of \emph{polarisation-preserving circuits} generated by beam splitters \input{bs-xs.tikz} and phase shifters
  \begin{tikzpicture}[scale=0.3]
	\begin{pgfonlayer}{nodelayer}
		\node [style=none] (5) at (-1.6, 0) {};
		\node [style=none] (9) at (1.6, 0) {};
		\node [style=PolRot,scale=0.75] (11) at (0, 0) {\small $\varphi$};
	\end{pgfonlayer}
	\begin{pgfonlayer}{edgelayer}
		\draw (5.center) to (11);
		\draw (11) to (9.center);
	\end{pgfonlayer}
\end{tikzpicture}
.
\end{definition}
Notice that we define polarisation-preserving circuits as a PRO rather than a PROP, thus they do not include swaps.

\subsection{Single-Photon Semantics}\label{sec:sem}

We will characterise photons by their spatial and polarisation modes.
Spatial modes refer to position, and polarisation can be horizontal (\H) or
vertical (\V).
Note that the quantum formalism admits (normalised complex) superpositions of
both spatial and polarisation modes.
For any
$n\in \mathbb N$, let $M_n= \{\V,\H\}\times[n]$, where
$[n]=\{0,\ldots n-1\}$, be the set of states (spatial and polarisation
modes).
The elements of $M_n$ are
denoted $c_p$ with $c\in \{\V,\H\}$ and $p\in [n]$.
The state space of
a single photon is
$\mathbb C^{{ M_n}} = span(\ket {\V_i},\ket{\H_i} ~|~ i\in
[n])$.
Notice that $\mathbb C^{{ M_0}} = \mathbb C^{\emptyset} =\{0\}$
is the Hilbert space of dimension $0$.
For instance, on 2 spatial modes (i.e. 2 wires), there are four
possible basis states: 
$\H_0, \H_1, \V_0, \V_1$. Indeed, a photon can be on
one of the two wires, while in the horizontal or
vertical polarisation. The state space is then a $4$-dimensional
Hilbert space.
The semantics of a \lov-circuit is defined as follows.

\begin{definition}
\label{singlephotsemantics}
For any \lov-circuit $D:n\to m$, let
$\interps D:\mathbb C^{{ M_n}}\to \mathbb C^{{ M_m}}$ be the linear
map inductively defined by 
Table~\ref{tab:singphotsem}\footnote{There are many possible conventions for
  beam splitters. We have chosen this one as it is a
  symmetric operation with
  good composition properties (see \cref{fig:usefulcons}). The convention for
  the wave plate has been chosen for similar reasons (see
  \cref{interestingequations}).}, and by
$\interps{D_2\circ D_1} = \interps{D_2} \circ \interps {D_1}$,
$\interps{D_1\oplus D_2} = \interps{D_1} \oplus \interps {D_2}$,
where for all
  $f \in\mathbb C^{{ M_n}}\to \mathbb C^{ M_m}$ and
  $g \in \mathbb C^{{ M_{n'}}}\to \mathbb C^{{ M_{m'}}}$,
  $(f\oplus g)(\ket{c_k}) = f(\ket{c_k})$ if $k<n$ and
  $S_{m,m'}(g(\ket{c_{k-n}}))$ if $k\ge n$, with
  $S_{m,m'}:\mathbb C^{{ M_{m'}}}\to   \mathbb C^{{ M_{m+m'}}} = \ket{c_k}\mapsto
  \ket{c_{k+m}}$ a shift of the positions by $m$.
\end{definition}

\begin{table}[t]
\begin{tabular}{@{}cc@{}}
$\begin{array}{rcl}
\interps{~\input{gene-0.tikz}~} &=& \interps{~\input{detector-0.tikz}~} ~=~ \interps{~\input{diagrammevide-s.tikz}~}  ~=~ 0\\[0.5cm]
\interps{~\input{bs-s.tikz}~}&=&\ket{c_p}\mapsto\cos(\theta)\ket{c_p}+i\sin(\theta)\ket{c_{1-p}}\\[0.3cm]
\interps{~\input{pol-rot.tikz}~} &=& \begin{cases} \ket{\V_0} \mapsto \cos(\theta) \ket{\V_0} + i\sin(\theta)\ket{\H_0} \\
\ket{\H_0} \mapsto  \cos(\theta) \ket{\H_0} + i\sin(\theta)\ket{\V_0} \end{cases}\
\end{array}$
&
$\begin{array}{rcl}
\interps{~\input{convtp-phase-shift-s.tikz}~} &=& \ket{c_0} \mapsto  e^{i\varphi} \ket{c_0}\\[.3cm]
\interps{~\input{beamsplitter-s.tikz}~} &=& \begin{cases} \ket{\V_p} \mapsto \ket{\V_p} \\ \ket{\H_p} \mapsto \ket{\H_{1-p}} 
 \end{cases}\\[.5cm]
\interps{~\input{swap-s.tikz}~} &=& \ket{c_p} \mapsto \ket{c_{1-p}}\\[.45cm]
\interps{~\input{filcourt-s.tikz}~} &=& \ket{c_0} \mapsto \ket{c_0}
\end{array}$
\end{tabular}
\caption{Semantics of \lov-circuits.}\label{tab:singphotsem}
\end{table}

\begin{example}
  The negation inverts polarisation:
  $\interps{~\begin{tikzpicture}[scale=0.8]
	\begin{pgfonlayer}{nodelayer}
		\node [style=cartouche] (0) at (0, 0) {\footnotesize$\neg$};
		\node [style=none] (1) at (-0.8, 0) {};
		\node [style=none] (2) at (0.8, 0) {};
	\end{pgfonlayer}
	\begin{pgfonlayer}{edgelayer}
		\draw [style=noire] (1.center) to (0);
		\draw [style=noire] (0) to (2.center);
	\end{pgfonlayer}
\end{tikzpicture}
~}: \ket{\V_0} \mapsto
    \ket{\H_0}$ and  $\ket{\H_0} \mapsto \ket{\V_0}$.
\end{example}

\begin{remark}
  The semantics of the circuits is sound with respect to
  the axioms of PROPs. In other words two circuits that are equal up to
  deformation have the same semantics. More formally,
  $\interps. : \lovbf \to (\textbf{Hilb}, \oplus, 0)$ is a monoidal
  functor where \textbf{Hilb} is the category of state spaces
  $\mathbb C^{{ M_n}}$ and linear maps.
\end{remark}

\begin{remark}
  All the generators of the \lov-circuits are photon preserving, even
  the vacuum state sources ($\begin{tikzpicture}[scale=0.7]
	\begin{pgfonlayer}{nodelayer}
		\node [style=none] (2) at (1.1, 0) {};
		\node [style=gene,scale=0.7] (3) at (0, 0) { $0$};
	\end{pgfonlayer}
	\begin{pgfonlayer}{edgelayer}
		\draw (3) to (2.center);
	\end{pgfonlayer}
\end{tikzpicture}
$) and detectors ($\begin{tikzpicture}[scale=0.7]
	\begin{pgfonlayer}{nodelayer}
		\node [style=none] (2) at (-1.1, 0) {};
		\node [style=detector,scale=0.7] (3) at (0, 0) {$0$};
	\end{pgfonlayer}
	\begin{pgfonlayer}{edgelayer}
		\draw (3) to (2.center);
	\end{pgfonlayer}
\end{tikzpicture}
$).
  Indeed the vacuum state source produces no
  photons, whereas the semantics of the detector corresponds to a
  postselection on the case where no photons are
  detected.
\end{remark}

\begin{definition}
  \label{def:Hsemantics}
For any \lopp-circuit $D:n\to n$, we define $\interph D:\mathbb C^n\to \mathbb C^n$ as the unique  linear map such that $\interp .  \circ \iota = \iota \circ \interph .$ where $\iota :\mathbb C^n \to \mathbb C^{M_n} = \ket k\mapsto \ket {\H_k}$. \end{definition}

For instance $\interph {\input{bs-xs.tikz}} = \left(\begin{smallmatrix}\cos(\theta)&i\sin(\theta)\\i\sin(\theta)&\cos(\theta)\end{smallmatrix}\right)$.

Polarisation-preserving circuits are universal for unitary transformations, this is a direct consequence of the result fo Reck \textit{et al.}~\cite{Reck1994unitary}. Unitary transformations can actually be uniquely represented by \lopp-circuits, as illustrated by the following two cases on 2 and 3 modes, the general case being proved in Section~\ref{section:pp}.

\begin{lemma}\label{existuniqtriangle2modes}
  For any unitary $2\times 2$ matrix $U$, there exist unique
  $\beta_1,\alpha_1\in[0,\pi)$ and $\beta_2,\beta_3\in[0,2\pi)$ such
  that $\interph{\scalebox{.8}{\input{phasebsphasebeta.tikz}}}=U$, 
   and
  $\alpha_1\in\{0,\frac\pi2\}\Rightarrow\beta_1=0$.
\end{lemma}
\begin{proof}
  The proof is given in \cref{preuveexistuniqtriangle2modes}.
\end{proof}

\begin{lemma}\label{existuniqtriangle3modes}
  For any unitary $3\times 3$ matrix $U$, there exist unique angles
  $\alpha_1,\alpha_2,\alpha_3,\beta_1,\beta_2,\beta_3\in[0,\pi)$ and
  $\beta_4,\beta_5,\beta_6\in[0,2\pi)$ such that
  $\interph{\scalebox{.8}{\input{bsyangbaxterpointeenhaut.tikz}}}=U$
  where
$\forall i\in\{1,2,3\},
\alpha_i\in\{0,\frac\pi2\}\Rightarrow\beta_i=0$, 
and where $\alpha_2=0\Rightarrow\alpha_1=0$.
\end{lemma}
\begin{proof}
  The existence of such a canonical form is shown in
  \cite{Reck1994unitary}. The uniqueness can then be derived by
  analysing the possible cases (See
  \cref{preuveexistuniqtriangle3modes}).
\end{proof}

$\lov$-circuits are more expressive than $\lopp$-ones, they not only act on the polarisation but the use of detectors and sources allow the representation of non-unitary evolutions:
For any $\lov$-circuit $D:n\to m$, $\interp D$ is sub-unitary\footnote{$U$ is sub-unitary (see for instance \cite{Sel2004-qpl}) iff $U^\dagger U \sqsubseteq I$, where $\sqsubseteq$ is the L\"owner partial order, i.e. $I-U^\dagger U$ is a positive semi-definite.}. $\lov$-circuits are actually universal for sub-unitary transformations:

\Scom{Déplacer le THM à la fin de la section 5?}
\begin{theorem}[Universality of \lov]
  For every sub-unitary map $U: \mathbb C^{M_n} \to \mathbb C^{M_m} $ (i.e. such that
  $U^\dagger U \sqsubseteq I$) there exists a diagram~$D:n\to m$ s.t. $\interp D =U$. 
\end{theorem}

\begin{proof} The proof given in Appendix \ref{appendix:universality} relies on the normal forms developed in Section \ref{subsect:completeness}. 
\end{proof}

\section{Equational Theory}
\label{section3}

Two distinct \lov-circuits may represent the same quantum
evolution: for instance, composing two negations is equivalent to
the identity.
In order to characterise equivalences of \lov-circuits, we
introduce a set of equations, shown in \cref{axiomsLOphotpres}.  They
capture basic properties of \lov-circuits, such as: detectors and
sources essentially absorbing the other generators (Equations
(\ref{absorptionphaseshiftleft}) to
(\ref{absorptionphaseshiftright})); parameters forming a monoid
(Equations (\ref{phaseaddition}) and (\ref{phase0})); and various
commutation properties (Equations (\ref{commutationphaserotpisur2}),
(\ref{gphpropapbs})).
Notice that there are two equations acting on 3 modes:
\cref{pbspbspbs} and \cref{Eulerscalaires}. \cref{pbspbspbs} is a
variant of the Yang-Baxter Equation \cite{jimbo1989introduction},
whereas \cref{Eulerscalaires} is a
property of decompositions into Euler angles.
Indeed, in 3-dimensional space, the two sides of
this equation
correspond to two distinct decompositions in elementary rotations.

\begin{definition}[$\lov$-calculus]
  \label{def:lov-calc}
  Two \lov-circuits $D_1, D_2$ are
  equivalent according to the rules of the \emph{$\lov$-calculus},
  denoted $\lov\vdash D_1=D_2$, if one can transform $D_1$ into $D_2$
  using the equations given in Figure \ref{axiomsLOphotpres}.  More
  precisely, $\lov\vdash \cdot = \cdot$ is defined as the smallest
  congruence which satisfies the equations of Figure
  \ref{axiomsLOphotpres} in addition to the axioms of PROP.  
\end{definition}

\begin{figure}[tb]
  \centering
\scalebox{.7}{\begin{minipage}{1.2\textwidth}
\begin{multicols}{2}
\newcommand{\nspazer}{-0.2em}
\begin{equation}\label{phaseaddition}\begin{array}{rcl}\input{convtp-phase-shifts-12.tikz}&=&\input{convtp-phase-shift-1plus2.tikz}\end{array}\end{equation}
\vspace{\nspazer}
\begin{equation}\label{phase0}\begin{array}{rcl}\input{phase-shift0.tikz}&=&\input{filcourt.tikz}\end{array}\end{equation}
\vspace{\nspazer}
\begin{equation}\label{bs0}\begin{array}{rcl}\input{bs0.tikz}&=&\input{filsparalleleslongbs-m.tikz}\end{array}\end{equation}
\vspace{\nspazer}
\begin{equation}\label{pbsnnnn}\begin{array}{rcl}\input{beamsplitternnnn-s.tikz}&=&\input{beamsplitterswap-s.tikz}\end{array}\end{equation}
\vspace{\nspazer}
\begin{equation}\label{pbspbs}\begin{array}{rcl}\input{beamsplitterbeamsplitter-s.tikz}&=&\input{filsparalleleslongs-s.tikz}\end{array}\end{equation}
\vspace{\nspazer}
\begin{equation}\label{pbspbspbs}\begin{array}{rcl}\input{bsbsbspointeenbas-s.tikz}&=&\input{xbsxpointeenbas-s.tikz}\end{array}\end{equation}
\vspace{\nspazer}
\begin{equation}\label{pbsnpbsh}\begin{array}{rcl}\input{beamsplitternhautbeamsplitter-s.tikz}&=&\input{beamsplitternnhaut-s.tikz}\end{array}\end{equation}
\vspace{\nspazer}
\begin{equation}\label{halterevide}\begin{array}{rcl}\input{haltere00.tikz}&=&\input{diagrammevide-s.tikz}\end{array}\end{equation}
\vspace{\nspazer}
\begin{equation}\label{absorptionphaseshiftleft}\begin{array}{rcl}\input{convtp-gene0phs.tikz}&=&\input{gene-0.tikz}\end{array}\end{equation}
\vspace{\nspazer}
\begin{equation}\label{absorptionpolrotleft}\begin{array}{rcl}\input{gene0polrot.tikz}&=&\input{gene-0.tikz}\end{array}\end{equation}
\vspace{\nspazer}
\begin{equation}\label{absorptionpbsleft}\begin{array}{rcl}\input{gene0pbs.tikz}&=&\input{gene0surgene0.tikz}\end{array}\end{equation}
\vspace{\nspazer}
\begin{equation}\label{absorptionphaseshiftright}\begin{array}{rcl}\input{convtp-phsdetector0.tikz}&=&\input{detector-0.tikz}\end{array}\end{equation}
\vspace{\nspazer}
\begin{equation}\label{absorptionpolrotright}\begin{array}{rcl}\input{polrotdetector0.tikz}&=&\input{detector-0.tikz}\end{array}\end{equation}
\vspace{\nspazer}
\begin{equation}\label{absorptionpbsright}\begin{array}{rcl}\input{pbsdetector0.tikz}&=&\input{detector0surdetector0.tikz}\end{array}\end{equation}
\vspace{\nspazer}
\begin{equation}\label{commutationphaserotpisur2}\begin{array}{rcl}\input{convtp-thetapisur2.tikz}&=&\input{convtp-pisur2theta.tikz}\end{array}\end{equation}
\vspace{\nspazer}
\begin{equation}\label{gphpropapbs}\begin{array}{rcl}\input{convtp-phshbpbs.tikz}&=&\input{convtp-pbsphshb.tikz}\end{array}\end{equation}
\vspace{\nspazer}
\begin{equation}\label{bsemulable}\begin{array}{rcl}\input{bs.tikz}&=&\input{bsemulation.tikz}\end{array}\end{equation}
\vspace{\nspazer}
\end{multicols}
\vspace{-0.5cm}
\begin{equation}\label{Eulerscalaires}\begin{array}{rcl}~\qquad\qquad\input{convtp-bsyangbaxterpointeenbas.tikz}&=&\input{bsyangbaxterpointeenhaut.tikz}\end{array}\end{equation}
\end{minipage}}
\caption{Axioms of the $\lov$-calculus. The equations are valid for arbitrary parameters $\varphi, \varphi_i, \theta,\theta_i \in \mathbb R$. 
In \cref{Eulerscalaires}, the angles on the left-hand side can take
any value while the right-hand side is given by \cref{existuniqtriangle3modes} (where $U$ is the $\interph.$-semantics of the left-hand side of the equation). 
\label{axiomsLOphotpres}}
\end{figure}

\begin{proposition}[Soundness]\label{soundnessLOphotpres}
  For any two \lov-circuits $D_1$ and
  $D_2$, if $\lov\vdash D_1 = D_2$ then $\interps{D_1}=\interps{D_2}$.
\end{proposition}
\begin{proof}
  Since semantic equality is a congruence it suffices to check that
  for every equation of \cref{axiomsLOphotpres} both sides have the
  same semantics, which follows from
  \cref{singlephotsemantics} and \cref{existuniqtriangle3modes}.
\end{proof}

\begin{proposition}\label{2piperiodic}
  The rules of the \lov-calculus imply that the parameters are
  $2\pi$-periodic, i.e. for any $\theta,\varphi\in\R$:
  \[\lov\vdash\!\!\! \scalebox{.71}{\input{bs.tikz}\!$=$\!\!\!   \input{bsthetaplus2pi.tikz}}\qquad\lov\vdash\!\!\! \scalebox{.71}{\begin{tikzpicture}
	\begin{pgfonlayer}{nodelayer}
		\node [style=none] (5) at (-1.1, 0) {};
		\node [style=none] (9) at (1.1, 0) {};
		\node [style=PolRot] (11) at (0, 0) {$\varphi$};
	\end{pgfonlayer}
	\begin{pgfonlayer}{edgelayer}
		\draw (5.center) to (11);
		\draw (11) to (9.center);
	\end{pgfonlayer}
\end{tikzpicture}
 \!$=$\!\!\!     \begin{tikzpicture}
	\begin{pgfonlayer}{nodelayer}
		\node [style=none] (5) at (-1.6, 0) {};
		\node [style=none] (9) at (1.6, 0) {};
		\node [style=PolRot] (11) at (0, 0) {$\varphi{+}2\pi$};
	\end{pgfonlayer}
	\begin{pgfonlayer}{edgelayer}
		\draw (5.center) to (11);
		\draw (11) to (9.center);
	\end{pgfonlayer}
\end{tikzpicture}
}\qquad\lov\vdash\!\!\! \scalebox{.71}{\begin{tikzpicture}
	\begin{pgfonlayer}{nodelayer}
		\node [style=none] (5) at (-1.1, 0) {};
		\node [style=none] (9) at (1.1, 0) {};
		\node [style=PhS] (11) at (0, 0) {};
		\node [style=none] (12) at (0.5, -0.4) {$\theta$};
	\end{pgfonlayer}
	\begin{pgfonlayer}{edgelayer}
		\draw (5.center) to (11);
		\draw (11) to (9.center);
	\end{pgfonlayer}
\end{tikzpicture}
 \!$=$\!\!\! 
    \begin{tikzpicture}
	\begin{pgfonlayer}{nodelayer}
		\node [style=none] (5) at (-1.1, 0) {};
		\node [style=none] (9) at (2.35, 0) {};
		\node [style=PhS] (11) at (0, 0) {};
		\node [style=none] (12) at (1.1, -0.4) {$\theta{+}2\pi$};
	\end{pgfonlayer}
	\begin{pgfonlayer}{edgelayer}
		\draw (5.center) to (11);
		\draw (11) to (9.center);
	\end{pgfonlayer}
\end{tikzpicture}
}\]
\end{proposition}
\begin{proof}
  The proof is given in \cref{preuve2piperiodic}.
\end{proof}

We now state one of our main results: the completeness of the \lov-calculus.
\begin{theorem}[Completeness]\label{completenessLOphotpres}
  For any two \lov-circuits $D_1$ and $D_2$, if $\interps{D_1}=\interps{D_2}$ then $\lov\vdash D_1 = D_2$.
\end{theorem}

The proof of \cref{completenessLOphotpres} is given in
\cref{subsect:completeness}. As a step towards proving the theorem,
we first consider the fragment of the \lopp-circuits.

\section{Polarisation-Preserving 
Circuits}
\label{section:pp}

This section gives a universal normal form for any
 \lopp-circuit. We prove the uniqueness of that
form by introducing a strongly normalising and confluent
polarisation-preserving rewrite system: \textup{PPRS}.

\begin{definition}
  \label{def:pprs}
  The rewrite system \textup{PPRS} is defined on
  \lopp-circuits with the rules of
  \cref{rulestriangleform}.
\end{definition}

\newcommand{\nspazer}{-0.2em}
\begin{figure}[tb]
\centering
\scalebox{0.9}{\begin{minipage}{\textwidth}
\begin{vwcol}[widths={0.5,0.5},
 sep=.1cm, justify=flush,rule=0pt,indent=0em] 
\begin{align}
  \label{phasemod2pi}\input{convtp-phase-shift-.tikz}\ &\to\ \input{convtp-phase-shiftthetamod2pi-.tikz}\\[1.5em]
    \label{bsmod2pi}\input{bs-.tikz}\ &\to\ \input{convtp-bsphimod2pi-.tikz}\\[1.5em]
      \label{fusionphaseshifts}\input{convtp-phase-shifts-12.tikz}\ &\to\ \input{convtp-phase-shift-1plus2.tikz}\\[1.5em]
         \label{zerophaseshifts}\input{convtp-phase-shift-zero.tikz}\ &\to\ \input{filcourt.tikz}\\[1.5em]
         \label{zerobs}\input{bs0.tikz}\ &\to \ \input{filsparalleleslongbs-m-.tikz}
  \end{align}

  \begin{align}
  \label{removebottomphase}\input{convtp-phasebbs.tikz}\ &\to\ \input{convtp-moinsthetahbsthetatheta.tikz}\\[0.6em]
  \label{passagepisur2}\input{convtp-phasehbspisur2.tikz}\ &\to\ 
\input{convtp-bspisur2thetab.tikz}\\[0.6em]
  \label{passagephasepi}\input{convtp-phasehbs.tikz}\ &\to\ \input{convtp-thetamoinspihbspimoinsphipib.tikz}\\[0.6em]
  \label{soustractionpi}
\input{bstheta4.tikz}\ &\to\ \input{bstheta4moinspipis.tikz}
  \end{align}
   \end{vwcol}
   
   \vspace{\nspazer}   \vspace{\nspazer}   \vspace{\nspazer} \vspace{\nspazer}
 
   \begin{equation}\label{glissadeEulerscalaires}\begin{array}{rcl}\input{convtp-bsyangbaxterpointeenbas-etoiles.tikz}&\to&\input{bsyangbaxterpointeenhaut.tikz}\end{array}\end{equation}
   
 \vspace{-0.4cm}

  \begin{equation}  \label{fusionEulerbsphasebs}\qquad\qquad\begin{array}{rcl}\input{bsphasebsalpha-etoile.tikz}&\to&\input{phasebsphasebeta.tikz}\end{array}\end{equation}
\end{minipage}}
  \caption{Rewriting rules of PPRS. 
    $\psi \in \mathbb R\setminus [0,2\pi)$, $ \varphi, \varphi_1, \varphi_2\in (0,2\pi)$, $\varphi_0,\theta_4\in [\pi, 2\pi)$, $\theta, \theta_0, \theta_1, \theta_2, \theta_3 \in (0,\pi)$, and $\theta_0\neq  \frac \pi 2$. 
  $\input{convtp-phase-shift-etoile.tikz}$ denotes either $\input{convtp-phase-shift.tikz}$ or $\input{filcourt.tikz}$. In Rules \eqref{glissadeEulerscalaires} and \eqref{fusionEulerbsphasebs}, the angles on the left-hand side can take any value while the right-hand side is given by \cref{existuniqtriangle3modes} and \cref{existuniqtriangle2modes} respectively.
  \label{rulestriangleform}}
\end{figure}

\begin{lemma}\label{soundnessrewriting}
  If $D_1$ rewrites to $D_2$ using the
  \textup{PPRS} rewrite system then $\lov\vdash D_1=D_2$.
\end{lemma}

\begin{proof}
  The proof is given in
  \cref{preuvesoundnessrewriting}.
\end{proof}

\begin{theorem}\label{thm:strongnormalisation}
  The rewrite system \textup{PPRS} is
  strongly normalising.
\end{theorem}

\begin{proof}
  The proof is done by defining a lexicographic order on six distinct
  values: numbers of beam splitters of various angle ranges, count of
  specific patterns, numbers and positions of phase shifters. The
  order is shown to be decreasing with respect to the rewrite rules of
  \textup{PPRS}. The complete proof is given in
  \cref{stronglynormalisingproof}.
\end{proof}

As \textup{PPRS} is terminating, we can therefore derive the
existence of normal forms. The next step is to show that these normal
forms are unique: this is derived from Theorem~\ref{lem:globconfluent}.

\begin{theorem}\label{lem:globconfluent}
  \textup{PPRS} is globally confluent.
\end{theorem}

\begin{proof}
  \textup{PPRS} is locally confluent. Indeed, one can show by case
  analysis that the non-trivial peaks all use at most three
  wires. Each peak can be closed since for any polarisation-preserving
  \lov-circuit of size $n\in \{1,2,3\}$, \textup{PPRS} terminates to a
  specific unique normal form: when $n=1$, a simple phase-shift; when
  $n=2$, the form shown in Lemma~\ref{existuniqtriangle2modes}; when
  $n=3$, the form shown in Lemma~\ref{existuniqtriangle3modes}.  See
  Appendix~\ref{app:proofnf123} for details.
  Finally, using \cref{thm:strongnormalisation}, global confluence is
  deduced from Newman's lemma~\cite{terese2003term}.
\end{proof}

\begin{figure}[t]
  \centering
  \input{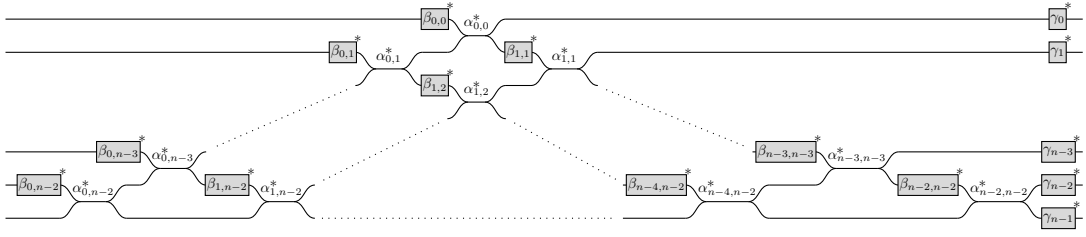}
  \caption{General scheme of a PPRS triangular normal form. The stars mean that any phase shifter or beam splitter with angle $0$ is replaced by the identity. The conditions on the angles are the following: $
  \alpha_{i,j},\beta_{i,j}\in[0,\pi)$;\quad $ 
  \gamma_i\in[0,2\pi)$;\quad $
  \alpha_{i,j}=0\Rightarrow\forall j'>j,\alpha_{i,j'}=0$;\quad $ 
  \alpha_{i,j}\in\{0,\frac\pi2\}\Rightarrow\beta_{i,j}=0$.}
  \label{fig:NFwithindices}
\end{figure}

\begin{figure}[t]
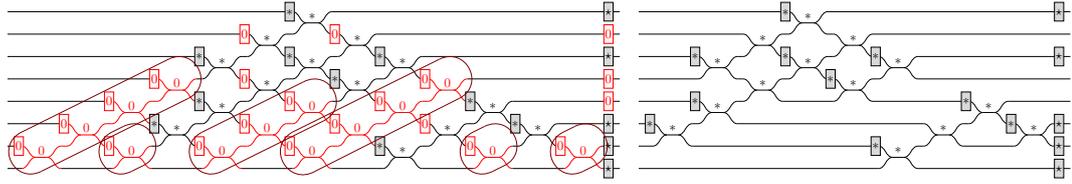

  \begin{subfigure}[b]{.6\linewidth}
    \centering
    \scalebox{0.55}{\input{triangleetoiles-8modes-star-s-diagonalesrouges-larges-zeros.tikz}}
  \end{subfigure}
  \begin{subfigure}[b]{.4\linewidth}
    \centering
    \scalebox{0.55}{\input{triangleetoiles-8modes-star-s-tronque-etroit.tikz}}
  \end{subfigure}
  \caption{An example of a PPRS triangular normal form. In the
    figure on the left, the beam splitters and phase shifters with angle $0$ in
    the corresponding triangular form are shown in red. In the
    figure on the right, they are replaced with identities.}
  \label{fig:nf}
\end{figure}

\begin{definition}\label{def:pprs-nf}
  A \emph{PPRS triangular normal form} is a circuit with a triangular
  shape similar to \cref{triangleNF}, 
but with all $0$-angled generators replaced with identities and with additional conditions on the angles, as described in \cref{fig:NFwithindices}.

\end{definition}

  Figure~\ref{fig:nf} shows an example: the figure on the left is the ``full''
  circuit with $0$-angled beam splitters while on the right is
  the corresponding PPRS triangular normal form.

\begin{lemma}\label{lem:normalformarePPRS}
  Any irreducible 
  \lopp-circuit is a PPRS triangular normal form.
\end{lemma}

\begin{proof}
  This property can be proven by induction. First, we explicit properties of any irreducible circuit that can be directly deduced from the PPRS rules of \cref{rulestriangleform}. 
  Then, we give two more properties characterising the PPRS triangular normal forms. 
  By induction, we prove that any irreducible circuit respects those two properties, so that any irreducible circuit is a PPRS triangular normal form. 
  The initialisation is deduced by \cref{lem:123normalform}. See \cref{preuvelemnormalformarePPRS} for more details.
\end{proof}
\begin{theorem}\label{thm:PPRSuniquenormalform}
  Any  
  \lopp-circuit, with the rules of
  \textup{PPRS}, converges to a unique \textup{PPRS} triangular normal form.
\end{theorem}

\begin{proof}
  \textup{PPRS} is globally confluent and terminating: normal forms
  are unique. From \cref{lem:normalformarePPRS}, PPRS triangular normal forms are the only irreducible forms. 
  Therefore, any
  polarisation-preserving circuit terminates to such a unique normal
  form.
\end{proof}

\begin{remark}
  By using Equation \eqref{Eulerscalaires} (together with Equations
  \eqref{phaseaddition} and \eqref{globalphasepropagationbs}) and by adding 0-angled beam splitters if necessary,  one can turn any circuit in PPRS triangular normal form into a circuit in the rectangular form of
  \cite{Clements2016unitary} shown in \cref{rectangleNF}. A schematic example of such a
  transformation is shown in \cref{triangletorectangularproof}. 

\end{remark}

We can now prove the completeness of the polarisation-preserving fragment.
\begin{theorem}\label{thm:ppcircuitscompleteness}
  For any 
  \lopp-circuits $C_1$,$C_2$ such that $\interph{C_1}=\interph{C_2}$, their normal forms are equal, i.e. $N_1=N_2$, where $N_1$ (resp. $N_2$) is the unique normal form of $C_1$ (resp. $C_2$) given by \cref{thm:PPRSuniquenormalform}.
\end{theorem}

\begin{proof}
  As the rewrite system preserves the semantics, it is sufficient to prove that $\interph{N_1}=\interph{N_2} \Rightarrow N_1=N_2$. 
  First, we can show by induction that $\interph{N}=\interph{I_n}\Rightarrow N=I_n$. Indeed, to have the semantics of identity, we can show the upper beam splitter and phase shifters are necessarily 0-angled. The proof follows from the induction property, details are given in \cref{preuvethmPPRScompleteness}.
  Let $P$ be an inverse circuit of $N_1$ and $N_2$, that is, a polarisation-preserving circuit such that $\interph{P}=\interph{N_1}^{-1}$. 
  The existence of such a circuit follows from \cite{Reck1994unitary}.
As $\interph{N_1P}=\interph{PN_2}=\interph{I_n}$, the term $N_1PN_2$ can both be reduced to $N_1$ (by reducing $PN_2$ first) and $N_2$ (by reducing $N_1P$ first). By \cref{thm:PPRSuniquenormalform}, $N_1=N_2$.
\end{proof}

\begin{proposition}[Universality and uniqueness in the polarisation-preserving fragment]\label{triangleuniversel}
  For any unitary $U\colon\mathbb C^{n}\to \mathbb C^{n}$, there
  exists a unique circuit $T$ in 
  \textup{PPRS} triangular normal form such that
  $\interph{T}=U$.
\end{proposition}

\begin{proof}
  This follows directly from \cite{Reck1994unitary}, \cref{thm:PPRSuniquenormalform,thm:ppcircuitscompleteness} and the fact that all PPRS triangular normal forms are irreducible. 
\end{proof}
  
\section{Completeness of the \texorpdfstring{\lov}{LOv}-Calculus}
\label{subsect:completeness}
To prove the completeness of the \texorpdfstring{\lov}{LOv}-Calculus (\cref{completenessLOphotpres}), we introduce the following
notion of normal form.

\begin{figure}[t]
  \centering
  \[\scalebox{0.8}{\input{NFexpl.tikz}}\]
  \caption{Shape of a circuit in normal form as of
    \cref{defNFLOphotpres}.}
  \label{fig:nf-lov-pr}
\end{figure}

\begin{definition}[Normal form]\label{defNFLOphotpres}
  A circuit in normal form $N:n\to m$ is a circuit of the form shown
  in \cref{fig:nf-lov-pr},
where $T$ is a PPRS triangular normal form (Definition~\ref{def:pprs-nf}).
If $n'=m'=0$,
then $N$ is said to be in \emph{pure normal form}.
\end{definition}

\begin{lemma}[Uniqueness of the pure normal form]\label{uniquenesspureNF}
  If two circuits $N_1$ and $N_2$ in pure normal form are such that
  $\interps{N_1}=\interps{N_2}$, then $N_1=N_2$.
\end{lemma}

\begin{proof}
Let $T_1$ (resp. $T_2$) be the \lopp-circuit associated with $N_1$ (resp $N_2$) as in \cref{fig:nf-lov-pr}. Notice that $\interph {T_i}\circ \mu = \mu \circ \interp{N_i}$ where $\mu: \mathbb C^{M_n}\to \mathbb C^{2n}$ is the isomorphism  $\ket{\V_k}\mapsto \ket{2k}$ and $\ket{\H_k}\mapsto \ket{2k+1}$. Thus $\interp{N_1} = \interp{N_2}$ implies $\interph{T_1}=\interph{T_2}$ so that the result
  follows from Theorem~\ref{thm:PPRSuniquenormalform}.
\end{proof}

\begin{lemma}\label{putinNF}
  For any circuit $D$ without vacuum state sources or detectors
  there exists a circuit in pure normal form
  $N$ such that $\lov\vdash D=N$.
\end{lemma}

\begin{proof}
The proof is given in
\cref{putinNFproof} 
\end{proof}

Completeness for circuits without vacuum state sources or detectors follows
directly from \cref{uniquenesspureNF,putinNF}:

\begin{proposition}\label{partialcompletenessLOphotpres}
  Given any two circuits $D_1$ and $D_2$ without any
  $\begin{tikzpicture}
	\begin{pgfonlayer}{nodelayer}
		\node [style=none] (2) at (1.1, 0) {};
		\node [style=gene] (3) at (0, 0) {$0$};
	\end{pgfonlayer}
	\begin{pgfonlayer}{edgelayer}
		\draw (3) to (2.center);
	\end{pgfonlayer}
\end{tikzpicture}
$ or $\begin{tikzpicture}
	\begin{pgfonlayer}{nodelayer}
		\node [style=none] (2) at (-1, 0) {};
		\node [style=detector] (3) at (0, 0) {$0$};
	\end{pgfonlayer}
	\begin{pgfonlayer}{edgelayer}
		\draw (3) to (2.center);
	\end{pgfonlayer}
\end{tikzpicture}
$, if
  $\interps{D_1}=\interps{D_2}$ then $\lov\vdash D_1=D_2$.
\end{proposition}

\begin{proof}
  By \cref{putinNF}, there exist two circuits in pure normal form
  $N_1$ and $N_2$ such that $\lov\vdash D_1=N_1$ and
  $\lov\vdash D_2=N_2$. By \cref{soundnessLOphotpres}, one has
  $\interps{N_1}=\interps{D_1}=\interps{D_2}=\interps{N_2}$, so that
  by \cref{uniquenesspureNF}, $N_1=N_2$. The result follows by
  transitivity.
\end{proof}

\paragraph*{Proof of \cref{completenessLOphotpres}}

We now have the required material to to finish the proof of
\cref{completenessLOphotpres}.  Let $D_1,D_2:n\to m$ be any two
\lov-circuits such that $\interps{D_1}=\interps{D_2}$. By deformation,
we can write them as
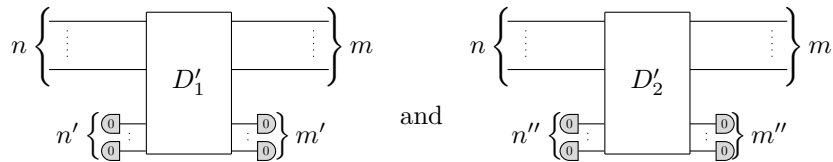
\[\input{blocsDprime1expl.tikz}\quad\text{and}\quad\input{blocsDprime2expl.tikz}\]
where $D'_1,D'_2$ do not contain $\begin{tikzpicture}
	\begin{pgfonlayer}{nodelayer}
		\node [style=none] (2) at (1.1, 0) {};
		\node [style=gene] (3) at (0, 0) {$0$};
	\end{pgfonlayer}
	\begin{pgfonlayer}{edgelayer}
		\draw (3) to (2.center);
	\end{pgfonlayer}
\end{tikzpicture}
$ or
$\begin{tikzpicture}
	\begin{pgfonlayer}{nodelayer}
		\node [style=none] (2) at (-1, 0) {};
		\node [style=detector] (3) at (0, 0) {$0$};
	\end{pgfonlayer}
	\begin{pgfonlayer}{edgelayer}
		\draw (3) to (2.center);
	\end{pgfonlayer}
\end{tikzpicture}
$. Up to using Equation \eqref{halterevide}, we
can assume that $n''=n'$. Since circuits without vacuum state sources and
detectors necessarily have the same number of input wires as of output
wires, this implies that $m''=m'$. By \cref{putinNF}, we can put
$D'_1$ and $D'_2$ in pure normal form. Then by using
\cref{absorptionphaseshiftleft,absorptionpolrotleft,absorptionpbsleft,absorptionphaseshiftright,absorptionpolrotright,absorptionpbsright},
we get two circuits in normal form
\[D_1^{\mathrm{NF}}=\input{NFblocs1majexpl.tikz}\text{ and }D_2^{\mathrm{NF}}=\input{NFblocs2majexpl.tikz}\]
with $T_1$ and $T_2$ in PPRS triangular normal form.

$\interp {D_1}=\interp {D_2}$ implies that $\pi\circ \interph{T_1} \circ \iota = \pi\circ \interph{T_2} \circ \iota$ where $\iota: \mathbb C^{2n}\to \mathbb C^{2n+n'}$ is the injection $\ket k\mapsto \ket k$ and $\pi : \mathbb C^{2m+m'}\to \mathbb C^{2m}$ is the projector s.t. $\pi\ket k =\ket k$ when $k<2m$ and $\pi\ket k=0$ otherwise. Thus there exists two unitaries $Q,Q'$ s.t. 
$\interph{T_2}=(I\oplus Q')\circ
\interph{T_1}\circ (I\oplus Q)$ (see \cref{egaliteaunitairespres} in Appendix \ref{egaliteaunitairespresproof}). 

By
\cref{triangleuniversel}, there exist two circuits $T_{\mathrm{in}}$
and $T_{\mathrm{out}}$ in PPRS triangular normal form such that $\interph{T_{\mathrm{in}}}=Q$ and $\interph{T_{\mathrm{out}}}=Q'$. By 
\cref{absorptionphaseshiftleft,absorptionphaseshiftright,absorptionbsleft,absorptionbsright}, we can make $T_{\mathrm{in}}$  and $T_{\mathrm{out}}$ appear, turning $D_1^{\mathrm{NF}}$ into
\[\input{NFblocs1majTTexpl.tikz}.\]
Since by construction, the middle part has the same single-photon semantics as $T_2$, by \cref{partialcompletenessLOphotpres}  
we can transform it into $T_2$ using the axioms of the $\lov$-calculus, which means transforming $D_1^{\mathrm{NF}}$ into $D_2^{\mathrm{NF}}$. The result follows by transitivity.
\qed

\section{Conclusion}
\label{conclusion}

In this paper, we presented the $\lov$-calculus, a graphical
language for LOQC capturing most of the components typically
considered in the Physics community for linear optical quantum
circuits. The language comes equipped with a sound and complete
semantics, and we discussed how it provides a unifying framework for
many of the existing approaches in the literature. We explained how
several existing results can be ported in the $\lov$ framework.

An obvious direction for future work is to extend the language to
allow for sources and detectors of a non-zero number of
photons. A more exploratory research avenue is to add support
for features such as squeezed states or continuous
variables.


\newpage
\appendix

\section{Proofs}

\subsection[2pi-Periodicity: Proof of Proposition \ref{2piperiodic}]{$2\pi$-Periodicity: Proof of Proposition \ref{2piperiodic}}\label{preuve2piperiodic}

We actually prove a stronger version of the $2\pi$-periodicity for the phase shifter:

\begin{equation}\label{phasemod2pieq}\begin{array}{rcl}\begin{tikzpicture}
	\begin{pgfonlayer}{nodelayer}
		\node [style=none] (5) at (-1.1, 0) {};
		\node [style=none] (9) at (1.1, 0) {};
		\node [style=PolRot] (11) at (0, 0) {$\varphi$};
	\end{pgfonlayer}
	\begin{pgfonlayer}{edgelayer}
		\draw (5.center) to (11);
		\draw (11) to (9.center);
	\end{pgfonlayer}
\end{tikzpicture}
&=&\begin{tikzpicture}
	\begin{pgfonlayer}{nodelayer}
		\node [style=none] (5) at (-2.1, 0) {};
		\node [style=none] (9) at (2.1, 0) {};
		\node [style=PolRot] (11) at (0, 0) {$\varphi\bmod 2\pi$};
	\end{pgfonlayer}
	\begin{pgfonlayer}{edgelayer}
		\draw (5.center) to (11);
		\draw (11) to (9.center);
	\end{pgfonlayer}
\end{tikzpicture}
\end{array}\end{equation}

as follows:

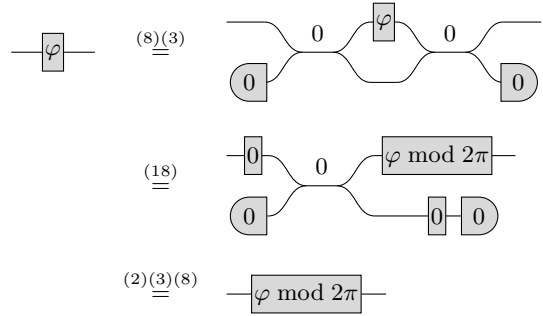
\begin{longtable}{RCL}
\begin{tikzpicture}
	\begin{pgfonlayer}{nodelayer}
		\node [style=none] (5) at (-1.1, 0) {};
		\node [style=none] (9) at (1.1, 0) {};
		\node [style=PolRot] (11) at (0, 0) {$\varphi$};
	\end{pgfonlayer}
	\begin{pgfonlayer}{edgelayer}
		\draw (5.center) to (11);
		\draw (11) to (9.center);
	\end{pgfonlayer}
\end{tikzpicture}
&\eqdeuxeqref{halterevide}{bs0}&\input{bs0phasephihbs0genedet.tikz}\\\\
&\eqeqref{Eulerscalaires}&\input{phasebsphasezerosphimod2pigenedet.tikz}\\\\
&\eqtroiseqref{phase0}{bs0}{halterevide}&\begin{tikzpicture}
	\begin{pgfonlayer}{nodelayer}
		\node [style=none] (5) at (-2.1, 0) {};
		\node [style=none] (9) at (2.1, 0) {};
		\node [style=PolRot] (11) at (0, 0) {$\varphi\bmod 2\pi$};
	\end{pgfonlayer}
	\begin{pgfonlayer}{edgelayer}
		\draw (5.center) to (11);
		\draw (11) to (9.center);
	\end{pgfonlayer}
\end{tikzpicture}

\end{longtable}

Then, the equality of \cref{2piperiodic} follows straightforwardly:
\[\begin{tikzpicture}
	\begin{pgfonlayer}{nodelayer}
		\node [style=none] (5) at (-1.1, 0) {};
		\node [style=none] (9) at (1.1, 0) {};
		\node [style=PolRot] (11) at (0, 0) {$\varphi$};
	\end{pgfonlayer}
	\begin{pgfonlayer}{edgelayer}
		\draw (5.center) to (11);
		\draw (11) to (9.center);
	\end{pgfonlayer}
\end{tikzpicture}
\ \eqeqref{phasemod2pieq}\ \begin{tikzpicture}
	\begin{pgfonlayer}{nodelayer}
		\node [style=none] (5) at (-2.1, 0) {};
		\node [style=none] (9) at (2.1, 0) {};
		\node [style=PolRot] (11) at (0, 0) {$\varphi\bmod 2\pi$};
	\end{pgfonlayer}
	\begin{pgfonlayer}{edgelayer}
		\draw (5.center) to (11);
		\draw (11) to (9.center);
	\end{pgfonlayer}
\end{tikzpicture}
\ \eqeqref{phasemod2pieq}\ \begin{tikzpicture}
	\begin{pgfonlayer}{nodelayer}
		\node [style=none] (5) at (-1.6, 0) {};
		\node [style=none] (9) at (1.6, 0) {};
		\node [style=PolRot] (11) at (0, 0) {$\varphi{+}2\pi$};
	\end{pgfonlayer}
	\begin{pgfonlayer}{edgelayer}
		\draw (5.center) to (11);
		\draw (11) to (9.center);
	\end{pgfonlayer}
\end{tikzpicture}
\]

To prove the $2\pi$-periodicity for the beam splitter, we proceed as follows:

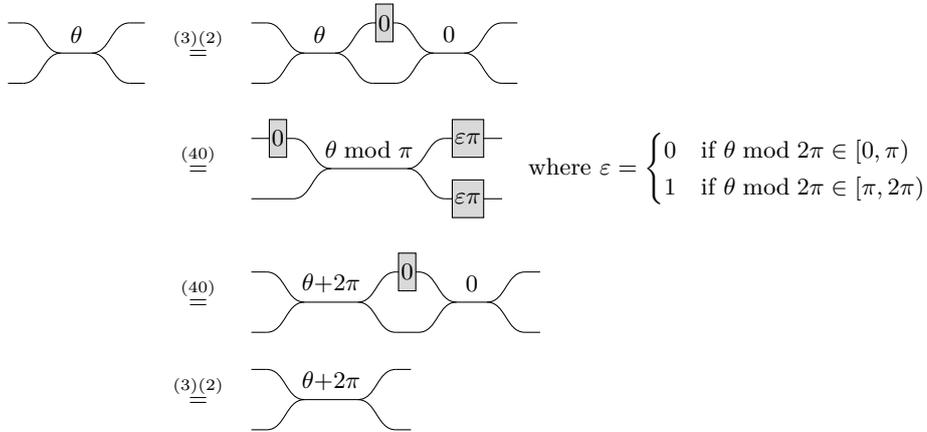
\begin{longtable}{RCL}
\input{bs.tikz}&\eqdeuxeqref{bs0}{phase0}&\input{convtp-bsphase0hbs0.tikz}\\\\
&\eqeqref{Eulerbsphasebs}&\input{phasebsphasethetamodpipis.tikz}\quad\text{where $\varepsilon=\begin{cases}0&\text{if $\theta\bmod 2\pi\in[0,\pi)$}\\1&\text{if $\theta\bmod 2\pi\in[\pi,2\pi)$}\end{cases}$}\\\\
&\eqeqref{Eulerbsphasebs}&\input{bsthetaplus2piphase0hbs0.tikz}\\\\
&\eqdeuxeqref{bs0}{phase0}&\input{bsthetaplus2pi.tikz}
\end{longtable}

Finally, the $2\pi$-periodicity for the wave plate follows from that for the beam splitter as follows:

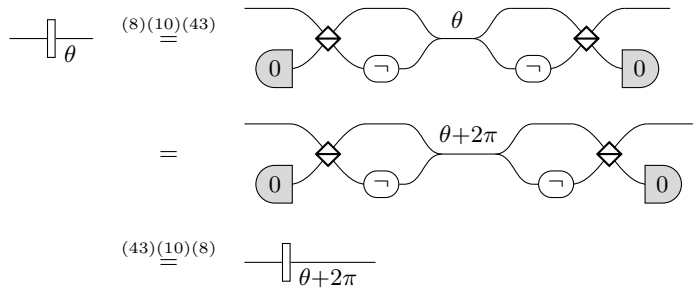
\begin{longtable}{RCL}
\begin{tikzpicture}
	\begin{pgfonlayer}{nodelayer}
		\node [style=none] (5) at (-1.1, 0) {};
		\node [style=none] (9) at (1.1, 0) {};
		\node [style=PhS] (11) at (0, 0) {};
		\node [style=none] (12) at (0.5, -0.4) {$\theta$};
	\end{pgfonlayer}
	\begin{pgfonlayer}{edgelayer}
		\draw (5.center) to (11);
		\draw (11) to (9.center);
	\end{pgfonlayer}
\end{tikzpicture}
&\eqtroiseqref{halterevide}{absorptionpolrotleft}{polrotsfrombs}&\input{pol-rot-NF.tikz}\\\\
&=&\input{pol-rotthetaplus2pi-NF.tikz}\\\\
&\eqtroiseqref{polrotsfrombs}{absorptionpolrotleft}{halterevide}&\begin{tikzpicture}
	\begin{pgfonlayer}{nodelayer}
		\node [style=none] (5) at (-1.1, 0) {};
		\node [style=none] (9) at (2.35, 0) {};
		\node [style=PhS] (11) at (0, 0) {};
		\node [style=none] (12) at (1.1, -0.4) {$\theta{+}2\pi$};
	\end{pgfonlayer}
	\begin{pgfonlayer}{edgelayer}
		\draw (5.center) to (11);
		\draw (11) to (9.center);
	\end{pgfonlayer}
\end{tikzpicture}

\end{longtable}

\begin{remark}\label{moduloborne}
Note that we could also prove the following stronger equations, for any $k\in\mathbb Z$, 
with the same sequence of rewrite steps, that is, in a bounded number of steps:

\[\begin{tikzpicture}
	\begin{pgfonlayer}{nodelayer}
		\node [style=none] (5) at (-1.1, 0) {};
		\node [style=none] (9) at (1.1, 0) {};
		\node [style=PolRot] (11) at (0, 0) {$\varphi$};
	\end{pgfonlayer}
	\begin{pgfonlayer}{edgelayer}
		\draw (5.center) to (11);
		\draw (11) to (9.center);
	\end{pgfonlayer}
\end{tikzpicture}
\ =\ \begin{tikzpicture}
	\begin{pgfonlayer}{nodelayer}
		\node [style=none] (5) at (-1.85, 0) {};
		\node [style=none] (9) at (1.85, 0) {};
		\node [style=PolRot] (11) at (0, 0) {$\varphi{+}2k\pi$};
	\end{pgfonlayer}
	\begin{pgfonlayer}{edgelayer}
		\draw (5.center) to (11);
		\draw (11) to (9.center);
	\end{pgfonlayer}
\end{tikzpicture}
\]

\[\input{bs.tikz}\ =\ \input{bsthetaplus2kpi.tikz}\]

\[\begin{tikzpicture}
	\begin{pgfonlayer}{nodelayer}
		\node [style=none] (5) at (-1.1, 0) {};
		\node [style=none] (9) at (1.1, 0) {};
		\node [style=PhS] (11) at (0, 0) {};
		\node [style=none] (12) at (0.5, -0.4) {$\theta$};
	\end{pgfonlayer}
	\begin{pgfonlayer}{edgelayer}
		\draw (5.center) to (11);
		\draw (11) to (9.center);
	\end{pgfonlayer}
\end{tikzpicture}
\ =\ \begin{tikzpicture}
	\begin{pgfonlayer}{nodelayer}
		\node [style=none] (5) at (-1.1, 0) {};
		\node [style=none] (9) at (2.85, 0) {};
		\node [style=PhS] (11) at (0, 0) {};
		\node [style=none] (12) at (1.325, -0.4) {$\theta{+}2k\pi$};
	\end{pgfonlayer}
	\begin{pgfonlayer}{edgelayer}
		\draw (5.center) to (11);
		\draw (11) to (9.center);
	\end{pgfonlayer}
\end{tikzpicture}
\]
\end{remark}

\subsection{Useful Consequences of the Axioms}
\label{usefulequationsproof}

\begin{lemma}
  The equations of Figure \ref{usefulequations} are consequences of
  the axioms of the $\lov$-calculus.
\end{lemma}

\begin{figure}[tbp]
\renewcommand{\nspazer}{-0.2em}
\centering
\scalebox{.769}{\begin{minipage}{1.3\textwidth}
\begin{multicols}{2}

\begin{equation}\label{Eulerbsphasebs}\begin{array}{rcl}\input{bsphasebsalpha.tikz}&=&\input{phasebsphasebeta.tikz}\end{array}\end{equation}
\vspace{\nspazer}
\begin{equation}\label{additionbs}\begin{array}{rcl}\input{convtp-bsphi1bsphi2.tikz}&=&\input{convtp-bsphi1plusphi2.tikz}\end{array}\end{equation}
\vspace{\nspazer}

\begin{equation}\label{globalphasepropagationbs}\begin{array}{rcl}\input{convtp-thetathetabs.tikz}&=&\input{convtp-bsthetatheta.tikz}\end{array}\end{equation}
\vspace{\nspazer}
\begin{equation}\label{polrotsfrombs}\begin{array}{rcl}\input{polrotsurpolrot.tikz}&=&\input{pbsnbbsthetanbpbs.tikz}\end{array}\end{equation}
\vspace{\nspazer}
\begin{equation}\label{negneg}\begin{array}{rcl}\input{negneg.tikz}&=&\input{filcourt.tikz}\end{array}\end{equation}
\vspace{\nspazer}
\begin{equation}\label{commutationnegphaseshift}\begin{array}{rcl}\input{convtp-negphaseshift.tikz}&=&\input{convtp-phaseshiftneg.tikz}\end{array}\end{equation}
\vspace{\nspazer}
\begin{equation}\label{absorptionnegleft}\begin{array}{rcl}\input{gene0neg.tikz}&=&\input{gene-0.tikz}\end{array}\end{equation}
\vspace{\nspazer}
\begin{equation}\label{absorptionnegright}\begin{array}{rcl}\input{negdetector0.tikz}&=&\input{detector-0.tikz}\end{array}\end{equation}
\vspace{\nspazer}
\begin{equation}\label{absorptionbsleft}\begin{array}{rcl}\input{convtp-gene0bs.tikz}&=&\input{gene0surgene0.tikz}\end{array}\end{equation}
\vspace{\nspazer}
\begin{equation}\label{absorptionbsright}\begin{array}{rcl}\input{convtp-bsdetector0.tikz}&=&\input{detector0surdetector0.tikz}\end{array}\end{equation}
\vspace{\nspazer}
\begin{equation}\label{bspisur2swap}\begin{array}{rcl}\input{bspisur2.tikz}&=&\input{swappissur2.tikz}\end{array}\end{equation}
\vspace{\nspazer}
\begin{equation}\label{bsswap}\begin{array}{rcl}\input{convtp-bsswap.tikz}&=&\input{convtp-swapbs.tikz}\end{array}\end{equation}
\vspace{\nspazer}
\begin{equation}\label{pbsswap}\begin{array}{rcl}\input{beamsplitterswap-s.tikz}&=&\input{swapbeamsplitter-s.tikz}\end{array}\end{equation}
\vspace{\nspazer}
\begin{equation}\label{bsemulablevariante}\begin{array}{rcl}\input{bs.tikz}&=&\input{convtp-bsemulationalenvers.tikz}\end{array}\end{equation}
\vspace{\nspazer}
\begin{equation}\label{commutationnegsbs}\begin{array}{rcl}\input{convtp-negnegbs.tikz}&=&\input{convtp-bsnegneg.tikz}\end{array}\end{equation}
\vspace{\nspazer}
\begin{equation}\label{commutationnegpolrot}\begin{array}{rcl}\input{negpolrot.tikz}&=&\input{polrotneg.tikz}\end{array}\end{equation}
\vspace{\nspazer}
\end{multicols}
\begin{equation}\label{echangebspbs}\begin{array}{rcl}\input{convtp-bsetpbscroises.tikz}&=&\input{convtp-pbsetbscroises.tikz}\end{array}\end{equation}
\end{minipage}}
\caption{\label{fig:usefulcons} Useful consequences of the axioms of the
  $\lov$-calculus. In Equation \eqref{Eulerbsphasebs}, the angles on the
  left-hand side can take any value, and the right-hand side is given by
  \cref{existuniqtriangle2modes}. Explicit expressions of $\alpha_1$, $\beta_1$, $\beta_2$ and 
  $\beta_3$ in terms of $\theta_1$, $\theta_2$ and $\varphi_1$ are
  given in Appendix \ref{preuveexistuniqtriangle2modes}.
\label{usefulequations}}
\end{figure}

  To prove \cref{Eulerbsphasebs}, we have:
  \begin{longtable}{RCL}
  \input{bsphasebsalpha.tikz}&\eqtroiseqref{halterevide}{bs0}{phase0}&\input{bsyangbaxterpointeenbaszerosgenedet.tikz}\\\\
  &\eqeqref{Eulerscalaires}&\input{bsyangbaxterpointeenhautzerosgenedet.tikz}\\\\
  &\eqtroiseqref{phase0}{bs0}{halterevide}&\input{phasebsphasebeta.tikz}
  \end{longtable}

  To prove \cref{additionbs}, we have:
  \begin{longtable}{RCL}
  \input{convtp-bsphi1bsphi2.tikz}&\eqeqref{phase0}&\input{convtp-bsphi1phase0hbsphi2.tikz}\\\\
  &\eqeqref{Eulerbsphasebs}&\input{phasebsphasebeta.tikz}\\\\
  &\eqeqref{Eulerbsphasebs}&\input{convtp-bsphi1plusphi2phase0hbs0.tikz}\\\\
  &\eqdeuxeqref{bs0}{phase0}&\input{convtp-bsphi1plusphi2.tikz}
  \end{longtable}

  To prove \cref{negneg}, we have (cf. \cite{alex2020pbscalculus}, Appendix D):
  \begin{longtable}{RCL}\label{preuvenegneg}
  \input{negneg.tikz}&\eqdeuxeqref{halterevide}{pbspbs}&\input{negnegsimplification1.tikz}\\\\
  &\eqeqref{pbsnnnn}&\input{negnegsimplification2.tikz}\\\\
  &\eqeqref{pbsnpbsh}&\input{negnegsimplification3.tikz}\\\\
  &\eqeqref{pbsnpbsh}&\input{negnegsimplification4.tikz}\\\\
  &\eqeqref{pbspbs}&\input{negnegsimplification5.tikz}\\\\
  &\eqeqref{pbsnnnn}&\input{negnegsimplification6.tikz}\\\\
  &\eqdeuxeqref{pbspbs}{halterevide}&\begin{tikzpicture}[scale=0.8]
	\begin{pgfonlayer}{nodelayer}
		\node [style=none] (0) at (-1.5, 0) {};
		\node [style=none] (1) at (1.5, 0) {};
	\end{pgfonlayer}
	\begin{pgfonlayer}{edgelayer}
		\draw [style=noire] (0.center) to (1.center);
	\end{pgfonlayer}
\end{tikzpicture}
  \end{longtable}
  
  Equation \eqref{commutationnegphaseshift} is a direct consequence of \cref{commutationphaserotpisur2,phaseaddition}.

  To prove \cref{globalphasepropagationbs}, we have:
  \begin{longtable}{RCL}
  \input{convtp-thetathetabs.tikz}&\eqdeuxeqref{halterevide}{bs0}&\input{convtp-bsyangbaxterpointeenbaszerosthetasgenedet.tikz}\\\\
  &\eqeqref{Eulerscalaires}&\input{bsyangbaxterpointeenhautzerosgenedetanglesoriginaux.tikz}\\\\
  &\eqeqref{Eulerscalaires}&\input{convtp-bsyangbaxterpointeenbaszerosthetasgenedet2.tikz}\\\\
  &\eqdeuxeqref{bs0}{halterevide}&\input{convtp-bsthetatheta.tikz}
  \end{longtable}

  To prove \cref{polrotsfrombs}, we have:
  \begin{longtable}{RCL}
  \input{polrotsurpolrot.tikz}&\eqdeuxeqref{pbspbs}{negneg}&\input{pbsnnbpbsprhprbpbsnnbpbs.tikz}\\\\
  &\eqeqref{bsemulable}&\input{pbsnbbsthetanbpbs.tikz}
  \end{longtable}
  
  Equation \eqref{absorptionnegleft} is a direct consequence of \cref{absorptionphaseshiftleft,absorptionpolrotleft}.
  
  Equation \eqref{absorptionnegright} is a direct consequence of \cref{absorptionphaseshiftright,absorptionpolrotright}.
  
  Equation \eqref{absorptionbsleft} is a direct consequence of \cref{bsemulable,absorptionnegleft,absorptionpbsleft,absorptionpolrotleft}.
  
  Equation \eqref{absorptionbsright} is a direct consequence of \cref{bsemulable,absorptionnegright,absorptionpbsright,absorptionpolrotright}.
  
  To prove \cref{bspisur2swap}, we have:
  \begin{longtable}{RCL}
  \input{bspisur2.tikz}&\eqeqref{bsemulable}&\input{bspisur2emulation.tikz}\\\\
  &\eqdeuxeqref{phase0}{phaseaddition}&\input{bspisur2emulationphasespisur2int.tikz}\\\\
  &=&\input{bspisur2emulationnegsphasespisur2int.tikz}\\\\
  &\eqdeuxeqref{gphpropapbs}{commutationnegphaseshift}&\input{bspisur2emulationnegsphasespisur2ext.tikz}\\\\
  &\eqeqref{negneg}&\input{bspisur2emulationnegsnnhphasespisur2ext.tikz}\\\\
  &\eqdeuxeqref{pbsnnnn}{pbspbs}&\input{nbswapnbpissur2.tikz}\\\\
  &\eqdeuxeqref{gphpropapbs}{commutationnegphaseshift}&\input{swappissur2.tikz}
  \end{longtable}
  
  To prove \cref{bsswap}, we have:
  \begin{longtable}{RCL}
  \input{convtp-bsswap.tikz}&\eqtroiseqref{phase0}{phaseaddition}{bspisur2swap}&\input{convtp-bsphibspisur2phasesmoinspisur2.tikz}\\\\
  &\eqtroiseqref{additionbs}{additionbs}{globalphasepropagationbs}&\input{convtp-bspisur2phasesmoinspisur2bsphi.tikz}\\\\
  &\eqtroiseqref{bspisur2swap}{phaseaddition}{phase0}&\input{convtp-swapbs.tikz}
  \end{longtable}
  
  Equation \eqref{pbsswap} is proved in \cite{alex2020pbscalculus} (Appendix B.2.2.1) as Equation (17).
  
  Equation \eqref{bsemulablevariante} is a direct consequence of \cref{bsemulable,bsswap,pbsswap}.
  
  To prove \cref{commutationnegsbs}, we have:
  \begin{longtable}{RCL}
  \input{convtp-negnegbs.tikz}&\eqeqref{bsemulable}&\input{convtp-negnegbsemulation.tikz}\\\\
  &\eqeqref{negneg}&\input{convtp-bsemulationmixtehb.tikz}\\\\
  &\eqeqref{negneg}&\input{convtp-bsemulationalenversnegneg.tikz}\\\\
  &\eqeqref{bsemulablevariante}&\input{convtp-bsnegneg.tikz}
  \end{longtable}
  
  To prove \cref{commutationnegpolrot}, we have:
  \begin{longtable}{RCL}
  \input{negpolrot.tikz}&\eqtroiseqref{halterevide}{absorptionpolrotleft}{polrotsfrombs}&\input{negpolrot1.tikz}\\\\
  &\eqdeuxeqref{absorptionnegleft}{negneg}&\input{negpolrot2.tikz}\\\\
  &\eqeqref{pbsnnnn}&\input{negpolrot3.tikz}\\\\
  &\eqeqref{bsswap}&\input{negpolrot4.tikz}\\\\
  &\eqdeuxeqref{pbsswap}{pbsnnnn}&\input{negpolrot5.tikz}\\\\
  &\eqdeuxeqref{negneg}{absorptionnegright}&\input{negpolrot6.tikz}\\\\
  &\eqtroiseqref{polrotsfrombs}{absorptionpolrotleft}{halterevide}&\input{polrotneg.tikz}
  \end{longtable}

  \textit{Proof of Equation \eqref{echangebspbs}.}
  To prove Equation \eqref{echangebspbs}, we need the following auxiliary equations, which are consequences of Equations \eqref{pbsnnnn}, \eqref{pbspbs}, \eqref{pbspbspbs} and \eqref{pbsnpbsh}:
  \begin{equation}\label{pbsdoubleneg}\input{pbsnhnb.tikz}=\input{swapnhnbpbs.tikz}\end{equation}
  \begin{equation}\label{pbsnegsbas}\input{beamsplitternnbas-s.tikz}=\input{swapnbpbsnh.tikz}\end{equation}
  \begin{equation}\label{naturalitebsswap}\input{bsxbspointeenbas-s.tikz}=\input{xbsbspointeenhaut-s.tikz}\end{equation}
  \begin{equation}\label{echangeswapsbas}\input{bsbsmonte-s.tikz}=\input{xbsbsxdescendbas-s.tikz}\end{equation}
  \begin{equation}\label{carre}\input{carrepbs-s.tikz}=\input{carrepbsswaps-s.tikz}\end{equation}
  \begin{equation}\label{bigebrepbs}\input{bigebrepbssansnegs.tikz}=\input{bigebresansnegsreduction5.tikz}\end{equation}
  \begin{equation}\label{bigebrepbsnegs}\input{bigebrepbsavecnegs.tikz}=\input{bigebrereduitedroite.tikz}\end{equation}
  \begin{equation}\label{bigebrepbsnegsswaps}\input{nhpbssurnbpbs.tikz}=\input{bigebrepbsavecnegsetswapsenplus.tikz}\end{equation}
  Equations \eqref{naturalitebsswap} and \eqref{echangeswapsbas} are proved in \cite{alex2020pbscalculus} (Appendix B.2.2.1) as Equations (25) and (24) respectively. Equations \eqref{pbsdoubleneg} and \eqref{pbsnegsbas} are direct consequences of Equations \eqref{pbsnnnn} and \eqref{negneg}.
  
  To prove Equation \eqref{carre}, we have:
  \begin{longtable}{RCL}
  \input{carrepbs-s.tikz}&\eqeqref{pbspbs}&\input{trianglehauttrianglebaspbs-s.tikz}\\\\
  &\eqeqref{pbspbspbs}&\input{trianglehauttrianglebaspbsswaps-s.tikz}\\\\
  &=&\input{carrepbsswaps-s.tikz}
  \end{longtable}
  to prove \cref{bigebrepbs}, we have:
  \begin{longtable}{RCL}
  \input{bigebrepbssansnegs.tikz}&\eqeqref{naturalitebsswap}&\input{bigebresansnegsreduction2.tikz}\\\\
  &\eqeqref{echangeswapsbas}&\input{bigebresansnegsreduction3.tikz}\\\\
  &=&\input{bigebresansnegsreduction4.tikz}\\\\
  &\eqeqref{carre}&\input{bigebresansnegsreduction5.tikz}
  \end{longtable}
  and to prove Equation \eqref{bigebrepbsnegs}, we have
  \begin{longtable}{RCL}
  \input{bigebrepbsavecnegs.tikz}&\eqeqref{negneg}&\input{bigebredoubleneghgbd.tikz}\\\\
  &\eqeqref{pbsdoubleneg}&\input{bigebrereduction1.tikz}\\\\

  &\eqeqref{bigebrepbs}&\input{bigebrereduction5.tikz}\\\\
  &=&\input{bigebrereduction6.tikz}\\\\
  &\eqeqref{pbsnegsbas}&\input{bigebrereduction7.tikz}\\\\
  &=&\input{bigebrereduitegauche.tikz}\\\\
  &=&\input{bigebrereduitedroite.tikz}
  \end{longtable}
  The last step is by mere deformation of the circuit, by exchanging the two PBS. To prove Equation \eqref{bigebrepbsnegsswaps}, we have:
  \begin{longtable}{RCL}
  \input{bigebrepbsavecnegsetswapsenplus.tikz}&\eqeqref{bigebrepbsnegs}&\input{bigebrepbsavecnegsetswapsenplusreduitegauche.tikz}\\\\
  &\eqeqref{negneg}&\input{nhpbssurnbpbs.tikz}
  \end{longtable}
  
  Now we can prove Equation \eqref{echangebspbs}:
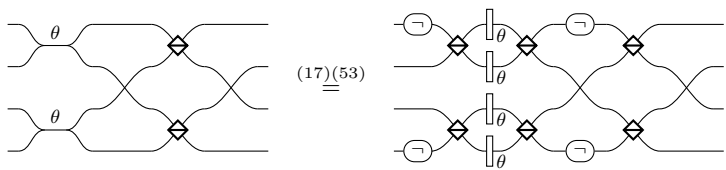
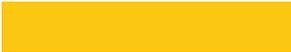
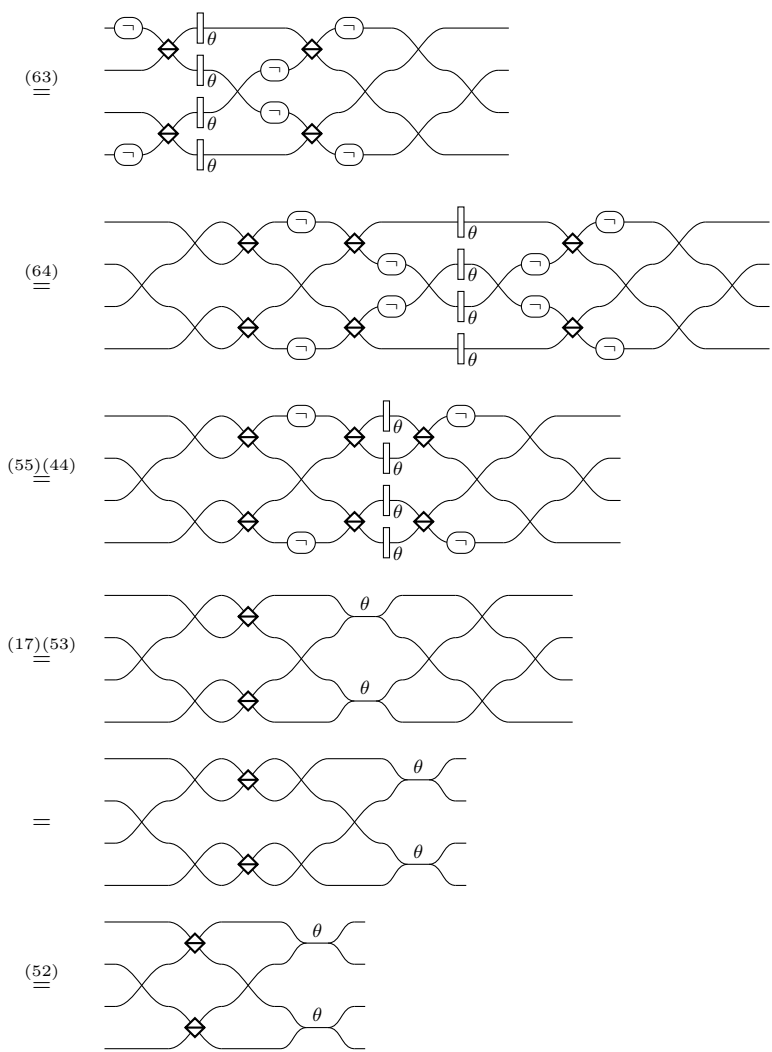
  \begin{longtable}{RCL}
  \input{convtp-bsetpbscroises-s.tikz}&\eqdeuxeqref{bsemulable}{bsemulablevariante}&\input{convtp-bsetpbscroisesbigebre1.tikz}\\\\
  &\eqeqref{bigebrepbsnegs}&\input{convtp-bsetpbscroisesbigebre2.tikz}\\\\
  &\eqeqref{bigebrepbsnegsswaps}&\input{convtp-bsetpbscroisesbigebre3.tikz}\\\\
  &\eqdeuxeqref{commutationnegpolrot}{negneg}&\input{convtp-bsetpbscroisesbigebre4.tikz}\\\\
  &\eqdeuxeqref{bsemulable}{bsemulablevariante}&\input{convtp-bsetpbscroisesbigebre5.tikz}\\\\
  &=&\input{convtp-bsetpbscroisesbigebre6.tikz}\\\\
  &\eqeqref{pbsswap}&\input{convtp-pbsetbscroises-s.tikz}
  \end{longtable}

\subsection{Triangular to Rectangular Form}
\label{triangletorectangularproof}

If necessary, we add 0-angled beam splitters in the PPRS triangular normal form to obtain a triangular shape, as in \cref{triangleNF}.
Then, for example with 7 spatial modes, we proceed as follows:\footnote{Here we only show how the beam splitters move along the process. We interpret Equation \eqref{Eulerscalaires} as sliding one of the beam splitters through the two others while changing the parameters and adding some phase shifts. Before and after each move it may be necessary to 
manipulate the phase shifters with the help of Equations \eqref{phaseaddition}, 
\eqref{globalphasepropagationbs} and \eqref{phase0}. Beam splitters represented in red are just to be moved, and beam splitters represented in blue have just been moved.}
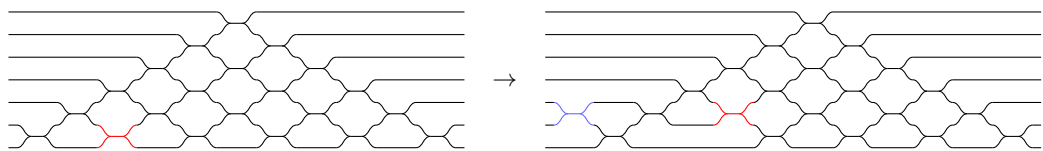
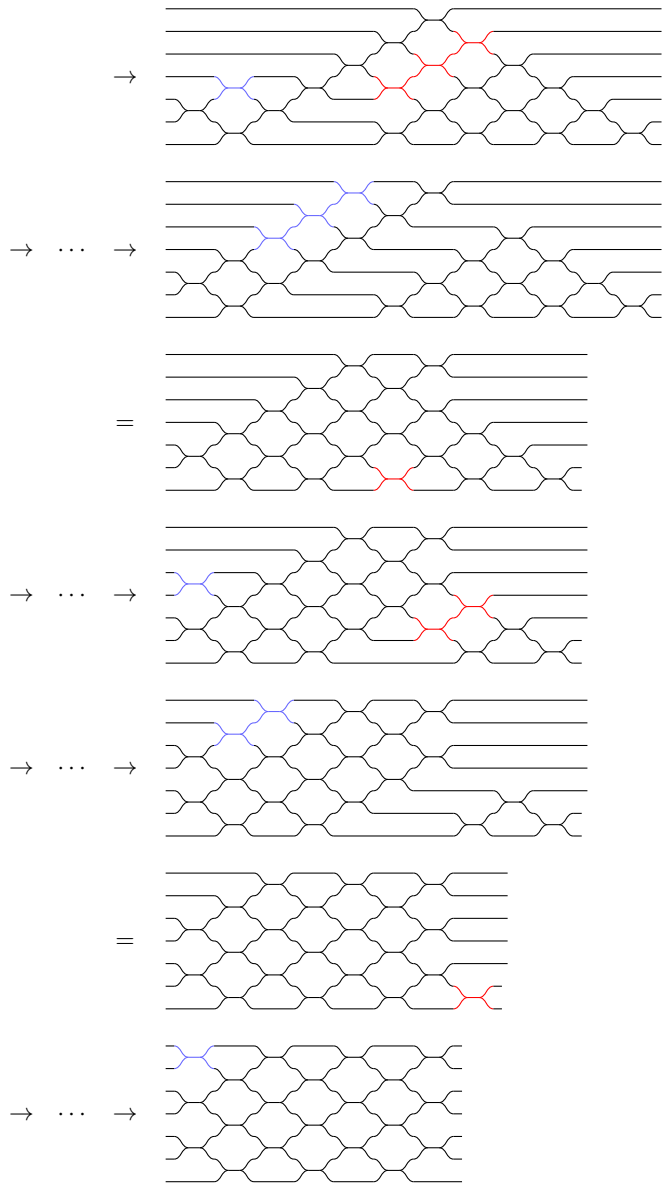
\begin{longtable}{RCL}
\input{triangleex1.tikz}&\to&\input{triangleex2.tikz}\\\\
&\to&\input{triangleex3.tikz}\\\\
\to\ \ \cdots&\to&\input{triangleex4.tikz}\\\\
&=&\input{triangleex6.tikz}\\\\
\to\ \ \cdots&\to&\input{triangleex7.tikz}\\\\
\to\ \ \cdots&\to&\input{triangleex8.tikz}\\\\
&=&\input{triangleex9.tikz}\\\\
\to\ \ \cdots&\to&\input{triangleex10.tikz}
\end{longtable}

This leads us to a rectangular form of \cite{Clements2016unitary} (see \cref{rectangleNF}).

\subsection{Existence and Uniqueness of the 3-Mode Triangular Form: Proof of Lemma \ref{existuniqtriangle3modes}}\label{preuveexistuniqtriangle3modes}

Let us consider such $\alpha_1,\alpha_2,\alpha_3,\beta_1,\beta_2,\beta_3\in[0,\pi)$ and $\beta_4,\beta_5,\beta_6\in[0,2\pi)$. We first prove that, assuming that they exist, their values are uniquely determined by $U$.

Let $U_1\coloneqq\interph{\input{filsurbeta1hbsalpha1.tikz}}\circ
U^\dag$,
$U_2\coloneqq\interph{\input{beta1alpha1beta2alpha2escalier.tikz}}\circ
U^\dag$ and
$U_3\coloneqq\interph{\input{bsyangbaxterpointeenhautsansphasesfinales.tikz}}\circ
U^\dag$, where $\interph{-}$ is defined in \cref{def:Hsemantics}.

By construction, $U_3=\begin{pmatrix}e^{-\ii\beta_4}&0&0\\0&e^{-\ii\beta_5}&0\\0&0&e^{-\ii\beta_6}\end{pmatrix}$, so that
\begin{eqnABC}\label{U2alphabeta}
U_2=\begin{pmatrix}e^{-\ii\beta_4}&0&0\\0&e^{-\ii(\beta_3+\beta_5)}\cos(\alpha_3)&-\ii e^{-\ii(\beta_3+\beta_6)}\sin(\alpha_3)\\0&-\ii e^{-\ii\beta_5}\sin(\alpha_3)&e^{-\ii\beta_6}\cos(\alpha_3)\end{pmatrix},
\end{eqnABC}
and $U_1=\interp{\input{beta2hbsalpha2surfil.tikz}}^\dag\circ U_2$. 
Since $\input{beta2hbsalpha2surfil.tikz}$ does not act on the last mode, this implies that $(U_1)_{2,0}=0$.\footnote{We denote by $M_{i,j}$ the entry of indices $(i,j)$ of a matrix $M$, the index of the first row and column being $0$.} That is, by definition of $U_1$, 
$\ii e^{\ii\beta_1}\sin(\alpha_1)U_{0,1}^\dag+\cos(\alpha_1)U_{0,2}^\dag=0$.

\begin{itemize}
\item If $U_{0,1},U_{0,2}\neq0$, then this equality implies that $\cos(\alpha_1)\neq0$ and $\sin(\alpha_1)\neq0$ (indeed, if $\cos(\alpha_1)=0$ then $\sin(\alpha_1)=\pm1$ and conversely, which in both cases prevents the equality from being satisfied). Hence, $\beta_1$ is the unique angle in $[0,\pi)$ such that $\frac{\ii e^{\ii\beta_1}U_{0,1}^\dag}{U_{0,2}^\dag}\in\R$, namely $\arg(U_{0,1})-\arg(U_{0,2})+\frac\pi2\bmod\pi$. Then $\alpha_1$ is the unique angle in $[0,\pi)\setminus\{\frac\pi2\}$ such that $\tan(\alpha_1)=-\frac{U_{0,2}^\dag}{\ii e^{\ii\beta_1}U_{0,1}^\dag}$.

\item If $U_{0,2}=0$ and $U_{0,1}\neq0$, then $\sin(\alpha_1)=0$, which means, since $\alpha_1\in[0,\pi)$, that $\alpha_1=0$. Due to the constraints on the angles, this implies that $\beta_1=0$ too.

\item If $U_{0,1}=0$ and $U_{0,2}\neq0$, then $\cos(\alpha_1)=0$, which means, since $\alpha_1\in[0,\pi)$, that $\alpha_1=\frac\pi2$. Due to the constraints on the angles, this implies that $\beta_1=0$ too.

\item If $U_{0,1}=U_{0,2}=0$, then since $U$ is unitary, it is of the form $U=\begin{pmatrix}e^{\ii\varphi}&0&0\\0&*&*\\0&*&*\end{pmatrix}$, where $*$ denotes any complex number. Then, regardless of $\alpha_1$ and $\beta_1$, $U_1$ is of the same form: $U_1=\begin{pmatrix}e^{\ii\varphi}&0&0\\0&*&*\\0&*&*\end{pmatrix}$. Consequently, $U_2=\begin{pmatrix}e^{\ii\varphi}\cos(\alpha_2)&*&*\\ \ii e^{\ii\varphi}\sin(\alpha_2)&*&*\\ 0&*&*\end{pmatrix}$. By \eqref{U2alphabeta}, this implies that $\sin(\alpha_2)=0$, which means, since $\alpha_2\in[0,\pi)$, that $\alpha_2=0$. Due to the constraints on the angles, this implies that $\alpha_1=\beta_1=\beta_2=0$ too.
\end{itemize}

\noindent Thus, $\alpha_1$ and $\beta_1$, and in turn $U_1$, are uniquely determined given $U$.

Since $(U_1)_{2,0}=0$, $U_1$ can be written as $\begin{pmatrix}(U_1)_{0,0}&*&*\\(U_1)_{1,0}&*&*\\0&*&*\end{pmatrix}$. By \eqref{U2alphabeta} we have $(U_2)_{1,0}=0$, that is, $\ii e^{\ii\beta_2}\sin(\alpha_2)(U_1)_{0,0}+\cos(\alpha_2)(U_1)_{1,0}=0$. Since $U_1$ is unitary, $|(U_1)_{0,0}|^2+|(U_1)_{1,0}|^2=1$, so that we cannot have $(U_1)_{0,0}=(U_1)_{1,0}=0$. The other cases are similar to those of $\alpha_1$ and $\beta_1$:

\begin{itemize}
\item If $(U_1)_{0,0},(U_1)_{1,0}\neq0$, then similarly, the equality implies that $\cos(\alpha_2)\neq0$ and $\sin(\alpha_2)\neq0$. Hence, $\beta_2$ is the unique angle in $[0,\pi)$ such that $\frac{\ii e^{\ii\beta_2}(U_1)_{0,0}}{(U_1)_{1,0}}\in\R$, namely $\arg((U_1)_{1,0})-\arg((U_1)_{0,0})+\frac\pi2\bmod\pi$. Then $\alpha_2$ is the unique angle in $[0,\pi)\setminus\{\frac\pi2\}$ such that $\tan(\alpha_2)=-\frac{(U_1)_{1,0}}{\ii e^{\ii\beta_2}(U_1)_{0,0}}$.

\item If $(U_1)_{1,0}=0$ and $(U_1)_{0,0}\neq0$, then $\sin(\alpha_2)=0$, which means, since $\alpha_2\in[0,\pi)$, that $\alpha_2=0$. Due to the constraints on the angles, this implies that $\beta_2=0$ too.

\item If $(U_1)_{0,0}=0$ and $(U_1)_{1,0}\neq0$, then $\cos(\alpha_2)=0$, which means, since $\alpha_2\in[0,\pi)$, that $\alpha_2=\frac\pi2$. Due to the constraints on the angles, this implies that $\beta_2=0$ too.
\end{itemize}

\noindent Thus, $\alpha_2$ and $\beta_2$, and in turn $U_2$, are also uniquely determined given $U$.

Furthermore, \eqref{U2alphabeta} implies that
\begin{itemize}
\item If $(U_2)_{1,1}, (U_2)_{2,1}\neq0$, then $\beta_3$ is the unique angle in $[0,\pi)$ such that $\frac{e^{\ii\beta_3}(U_2)_{1,1}}{\ii(U_2)_{2,1}}\in\R$, namely, $\arg((U_2)_{2,1})-\arg((U_2)_{1,1})+\frac\pi2\bmod\pi$, and $\alpha_3$ is the unique angle in $[0,\pi)$ such that $\tan(\alpha_3)=\frac{\ii(U_2)_{2,1}}{e^{\ii\beta_3}(U_2)_{1,1}}$.
\item If $(U_2)_{2,1}=0$ and $(U_2)_{1,1}\neq0$ then $\sin(\alpha_3)=0$, which means, since $\alpha_3\in[0,\pi)$, that $\alpha_3=0$. Due to the constraints on the angles, this implies that $\beta_3=0$ too.
\item If $(U_2)_{1,1}=0$ and $(U_2)_{2,1}\neq0$, then $\cos(\alpha_3)=0$, which means, since $\alpha_3\in[0,\pi)$, that $\alpha_3=\frac\pi2$. Due to the constraints on the angles, this implies that $\beta_3=0$ too.
\end{itemize}

\noindent Thus, $\alpha_3$ and $\beta_3$, and in turn $U_3$, are also uniquely determined given $U$.

Finally, since $U_3=\begin{pmatrix}e^{-\ii\beta_4}&0&0\\0&e^{-\ii\beta_5}&0\\0&0&e^{-\ii\beta_6}\end{pmatrix}$, we necessarily have $\beta_4=-\arg((U_3)_{0,0})$, $\beta_5=-\arg((U_3)_{1,1})$ and $\beta_6=-\arg((U_3)_{2,2})$. This finishes proving the uniqueness.

Conversely, it is easy to see that the unique possible values given above for $\alpha_1$, $\alpha_2$, $\alpha_3$, $\beta_1$, $\beta_2$, $\beta_3$, $\beta_4$, $\beta_5$ and $\beta_6$ are well defined for any unitary $U$ and satisfy the desired properties, which proves the existence.

\begin{remark}
It is possible to generalise this proof to extend the result to an arbitrary number of modes. 
This provides an alternative proof of 
\cref{triangleuniversel}.
\end{remark}

\subsection{Existence and Uniqueness of the 2-Mode Triangular Form: Proof of Lemma \ref{existuniqtriangle2modes}}\label{preuveexistuniqtriangle2modes}

Let us consider such $\beta_1,\alpha_1\in[0,\pi)$ and $\beta_2,\beta_3\in[0,2\pi)$. We have (with $\interph{-}$ defined in \cref{def:Hsemantics}):

\[U=\interph{\input{phasebsphasebeta.tikz}}=\begin{pmatrix}e^{\ii(\beta_1+\beta_2)}\cos(\alpha_1)&\ii e^{\ii\beta_2}\sin(\alpha_1)\\\ii e^{\ii(\beta_1+\beta_3)}\sin(\alpha_1)&e^{\ii\beta_3}\cos(\alpha_1)\end{pmatrix}\]

If $U$ has a null entry, the since it is unitary, it is either diagonal or anti-diagonal. If it is diagonal, then $\sin(\alpha_1)=0$, which, since $\alpha_1\in[0,\pi)$, implies that $\alpha_1=0$, which by the constraint on $\beta_1$ and $\alpha_1$, implies that $\beta_1=0$. Consequently, $\beta_2=\arg(U_{0,0})$ and $\beta_3=\arg(U_{1,1})$. If $U$ is anti-diagonal, then $\cos(\alpha_1)=0$, which, since $\alpha_1\in[0,\pi)$, implies that $\alpha_1=\frac\pi2$, which by the constraint on $\beta_1$ and $\alpha_1$, implies that $\beta_1=0$. Consequently, $\beta_2=
\arg(\frac{U_{0,1}}{\ii})$ and $\beta_3=
\arg(\frac{U_{1,0}}{\ii})$.

If $U$ has no null entry, since $UU^\dag=I$, we have $e^{\ii(\beta_1+\beta_2)}\cos(\alpha_1)U_{1,0}^\dag+\ii e^{\ii\beta_2}\sin(\alpha_1)U_{1,1}^\dag=0$. Hence, $\beta_1$ is the unique angle in $[0,\pi)$ such that $\frac{e^{\ii\beta_1}U_{1,0}^\dag}{\ii U_{1,1}^\dag}\in\R$, namely $\arg(U_{1,0})-\arg(U_{1,1})+\frac\pi2\bmod\pi$. Then $\alpha_1$ is the unique angle in $[0,\pi)$ such that $\tan(\alpha_1)=-\frac{e^{\ii\beta_1}U_{1,0}^\dag}{\ii U_{1,1}^\dag}$, and since $\alpha_1\in(0,\pi)$, we have $\sin(\alpha_1)>0$, so that $\beta_2=\arg{\frac{U_{0,1}}{\ii}}$ and $\beta_3=
\arg(\frac{U_{1,0}}{\ii e^{\ii\beta_1}})$.

Finally, it is easy to see that given any unitary $U$, the unique possible values given above for $\beta_1$, $\alpha_1$, $\beta_2$ and $\beta_3$ are well defined and satisfy the desired properties.

\subsection{Strong Normalisation: Proof of Lemma \ref{thm:strongnormalisation}}

\label{stronglynormalisingproof}

  To prove that the rewrite system is strongly normalising, given a circuit $D:n\to n$ composed only of phase shifters and beam splitters, let us consider the tuple $(a,b,c,d,e,f)$, where:
  \begin{itemize}
  \item $a$ is the number of beam splitters in $D$ with angle not in $[0,\pi)$
  \item $b$ is the number of beam splitters in $D$ with angle not in $[0,2\pi)$
  \item $c=\sum_{i=0}^{n-2}c(i)\times2^{n-i}$ where $c(i)$ is the number of beam splitters in $D$ 
that act on positions $i$ and $i+1$\alexandrecom{pas besoin de l'exponentiation}

  \item $d=(d_j)_{j\in\mathbb N}$ is the almost-zero sequence of integers such that for any $j\geq0$,
  \begin{itemize}
  \item $d_{2j}$ is the number of phase shifters $\begin{tikzpicture}
	\begin{pgfonlayer}{nodelayer}
		\node [style=PolRot] (11) at (0, 0) {$\varphi$};
	\end{pgfonlayer}
\end{tikzpicture}
$ of depth $j$ that are in a pattern 
  $\input{convtp-phasehbs-chaine.tikz}$\medskip with either 
  $\theta
=\frac\pi2$, or 
  $\theta\in(0,\pi)\setminus\{\frac{\pi}{2}\}$ and there exists a subset $I$ of $\{1,...,k\}$ such that $(\varphi+\sum_
{i\in I}\varphi_i)\bmod 2\pi\notin[0,\pi)$
  \item $d_{2j+1}$ is the number of phase shifters $\begin{tikzpicture}
	\begin{pgfonlayer}{nodelayer}
		\node [style=PolRot] (11) at (0, 0) {$\varphi$};
	\end{pgfonlayer}
\end{tikzpicture}
$ of depth $j$ that are in a pattern $\input{convtp-phasebbs-chaine.tikz}$\medskip
  \end{itemize}
  where the depth of a phase shifter $p$ in $D$ is defined as the maximal number of beam splitters that a photon starting from $p$ and going to the right can traverse before reaching an output port
  \item $e$ is the number of phase shifters in $D$
  \item $f$ is the number of phase shifters in $D$ with angle not in $[0,2\pi)$.
  \end{itemize}
  We consider the order on almost-zero sequences $(d_j)_{j\in\mathbb N}$ given, for two such sequences, by comparing first the indices of their respective last non-zero terms, then in case of equality, the terms with greatest index with different values. It is well known that this gives a well-order, more precisely of order type $\omega^\omega$. 
We consider the lexicographic order on the set of tuples $(a,b,c,d,e,f)$. Since it is the lexicographic order on a finite cartesian product of well-ordered sets, it is itself a well-order. 
 Hence, to prove that the rewrite system is strongly normalising, it suffices to prove that each of the rewrite rules strictly decreases 
  the tuple $(a,b,c,d,e,f)$.
  \begin{itemize}
  \item Rule \eqref{phasemod2pi} strictly decreases $f$ without increasing any component of the tuple.
  \item Rule \eqref{bsmod2pi} strictly decreases $b$ without increasing $a$.
  \item Rule \eqref{fusionphaseshifts} strictly decreases $e$, it does not change $a$, $b$  
or $c$ since it does not affect the beam splitters, and it does not increase $d$. Indeed, 
if after applying Rule \eqref{fusionphaseshifts}, $\begin{tikzpicture}
	\begin{pgfonlayer}{nodelayer}
		\node [style=PolRot] (15) at (0, 0) {$\varphi_1{+}\varphi_2$};
	\end{pgfonlayer}
\end{tikzpicture}
$ is at the left of a sequence of phase shifters 
that contains a subsequence with sum of angles modulo $2\pi$ 
not in $[0,\pi)$, then this was already the case of $\begin{tikzpicture}
	\begin{pgfonlayer}{nodelayer}
		\node [style=PolRot] (18) at (0, 0) {$\varphi_1$};
	\end{pgfonlayer}
\end{tikzpicture}
$ before applying Rule~\eqref{fusionphaseshifts}.
  \item Rule~\eqref{zerophaseshifts} does not increase $a$, $b$ 
or $c$ since it does not affects the beam splitters, it does not increase $d$ since it only removes a phase shifter, and it strictly decreases $e$.
  \item Rule~\eqref{zerobs} does not increase 
$a$ or $b$ since it only removes a beam splitter, and strictly decreases 
$c$.
  \item Rule \eqref{removebottomphase} does not change $a$, $b$ 
or $c$ since it does not affect the beam splitters, and it strictly decreases $d$. Indeed, it decreases by $1$ the number of phase shifters at the bottom left of a beam splitter at a given depth $j$, which decreases $d_{2j+1}$, and 
  it only adds phase shifters at the top left of this same beam splitter, which possibly increases only $d_{2j}$, and at the right of this beam splitter, that is, at depth
at most $j-1$, which can only change the $d_{j'}$ with $j'\leq 2j-1$.

  \item Rule \eqref{passagepisur2} does not change $a$, $b$ 
or $c$ since it does not affect the beam splitters, and it strictly decreases $d$. Indeed, it decreases by $1$ the number of 
phase shifters at the top left of a $\frac{\pi}{2}$-angled beam splitter at a given depth $j$, which decreases $d_{2j}$, and adds a phase shifter at the right of this same beam splitter, that is, at depth at most $j-1$, which can only change the $d_{j'}$ with $j'\leq 2j-1$.
  \item Rule \eqref{passagephasepi} does not change $a$, $b$ 
or $c$ since it affects the beam splitters only by changing the angle of one of them and keeps this angle in $(0,\pi)$ (and in particular, different from $0$), and it strictly decreases $d$. Indeed, for some beam splitter $\input{bs.tikz}$ with $\theta\in(0,\pi)\setminus\{\frac{\pi}{2}\}$, it decreases by at least $1$ the number of phase shifters $\begin{tikzpicture}
	\begin{pgfonlayer}{nodelayer}
		\node [style=PolRot] (11) at (0, 0) {$\varphi$};
	\end{pgfonlayer}
\end{tikzpicture}
$ belonging to a pattern $\input{convtp-phasehbs-chaine.tikz}$ formed with this beam splitter such that for some $I\subseteq\{1,...,k\}$, $(\varphi+\sum_
{i\in I}\varphi_i)\bmod 2\pi\notin[0,\pi)$, and it does not turn $\theta$ into $\frac{\pi}{2}$. The affected phase shifter is at a given depth $j$, so that applying this rule decreases $d_{2j}$. Additionally, the application of Rule~\eqref{passagephasepi} adds a phase shifter at the right of the beam splitter, that is, at depth at most $j-1$, which can only change the $d_{j'}$ with $j'\leq 2j-1$.
  \item Rule \eqref{soustractionpi} strictly decreases $a$.
  \item Rule \eqref{glissadeEulerscalaires} strictly decreases $c$, and it does not increase $a$ or $b$ since it only outputs beam splitters with angle in $[0,\pi)$.

  \item Rule \eqref{fusionEulerbsphasebs} does not increase $a$ or $b$ since it can only output a beam splitter with angle in $[0,\pi)$, and

  it strictly decreases 
  $c$.
  \end{itemize}

\subsection{Proof of Local Confluence}
\label{app:proofnf123}

\begin{lemma}\label{lem:123normalform}
  For any polarisation-preserving \lov-circuit of size $n\in \{1,2,3\}$,
  \textup{PPRS} terminates to a unique normal form 
  with the shape shown in Figure~\ref{tab:PPRS-NF}.
\end{lemma}

\begin{figure}[t]
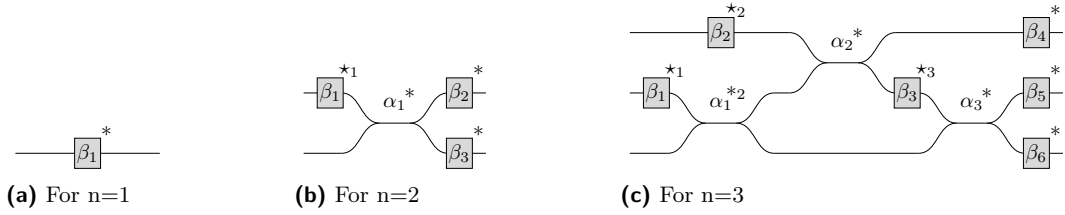

    \captionsetup[subfigure]{position=b}
    \subcaptionbox{\label{PPRS-NF1}For
      n=1}{\scalebox{.8}{\input{PPRS-NF1.tikz}}}
    \hfill
    \subcaptionbox{\label{PPRS-NF2}For
      n=2}{\scalebox{.8}{\input{PPRS-NF2.tikz}}}
    \hfill
    \subcaptionbox{\label{PPRS-NF3}For n=3}{\scalebox{.8}{\input{PPRS-NF3.tikz}}}
    \caption{Normal forms of PPRS for
      $n\in\{1,2,3\}$. $*$ means that the phase shifter or beam splitter is
      replaced by (an) identity wire(s) when the angle is
      zero. 
      $*_i$ represents the identity in the preceding case and also when $\alpha_i = 0$.  
      $\star_i$ represents the identity in the preceding two cases and also when $\alpha_i=\frac{\pi}{2}$. The $\alpha_i$ are in $[0,\pi)$ as well as the phases with a $\star_i$, all other phases are in $[0,2\pi)$.}\label{tab:PPRS-NF}
\end{figure}

\begin{proof}
First, we will show that the normal forms are necessarily of the form given in \cref{tab:PPRS-NF}.

In a normal form, because of Rule~\eqref{phasemod2pi}, all phase shifters have angle in $[0,2\pi)$; because of Rules \eqref{bsmod2pi} and \eqref{soustractionpi}, all beam splitters have angle in $[0,\pi)$; because of Rules \eqref{zerophaseshifts} and \eqref{zerobs}, there is no phase shifters or beam splitters with angle $0$; because of Rule~\eqref{passagepisur2}, there is no phase shifter at the top left of a $\frac\pi2$-angled beam splitter; and because of Rule~\eqref{passagephasepi}, all phase shifters at the top left of a beam splitter have angle in $[0,\pi)$. Thus, if a normal form is of one of the three forms given in \cref{tab:PPRS-NF}, then the conditions on the angles are satisfied.

Because of Rule~\eqref{fusionphaseshifts}, a normal form cannot contain two consecutive phase shifters. This implies in particular that the normal forms have the claimed shape for $n=1$.

Additionally, a normal form also cannot contain two consecutive beam splitters separated only by phase shifters (i.e., a pattern of the form $\input{bsphaseshbbsalpha-etoiles.tikz}$). Indeed, because of Rule~\eqref{removebottomphase}, in such a pattern in a normal form, there would not be a phase shifter on the bottom wire, so that the pattern would be reducible by Rule~\eqref{fusionEulerbsphasebs}. Thus, in the case where $n=2$, a normal form contains at most one beam splitter. Because of Rule~\eqref{removebottomphase}, such a beam splitter does not have any phase shifter at its bottom left, and because of Rule~\eqref{fusionphaseshifts}, there is at most one phase shifter on each of its other three ports. Because of Rule~\eqref{passagepisur2}, there is no phase shifter on the bottom right if the angle of the beam splitter is $\frac\pi2$. Moreover, if the normal form does not contain a beam splitter, then because of Rule~\eqref{fusionphaseshifts} there is at most one phase shifter on each of the two wires. Thus, in all cases, the normal forms have the claimed shape for $n=2$.

In the case where $n=3$, since there cannot be two consecutive beam splitters in a normal form, the beam spitters are alternatively between the top two wires and the bottom two wires. Because of Rules \eqref{removebottomphase} and \eqref{fusionphaseshifts}, if there is a beam splitter between the top two wires, then one betwen the bottom two wires, and then again one between the top two wires, those three beam splitter necessarily match the left-hand side of Rule~\eqref{glissadeEulerscalaires}. Hence, a normal form contains at most three beam splitters, at most one on the top and two on the bottom. Additionally, if the one on the top is not here, then because of Rule~\eqref{fusionEulerbsphasebs} there is only one on the bottom. Finally, Rules \eqref{removebottomphase} and \eqref{fusionphaseshifts} guarantee that the phase shifters are such that the normal form has the shape given in \cref{PPRS-NF3}.

  Now, we want to show that those are unique forms. 
One can check that given any circuit of the form given in \cref{PPRS-NF2} (resp. \cref{PPRS-NF3}), there is a unique way of adding $0$-angled phase shifters and beam splitters that gives us a circuit of the form of \cref{existuniqtriangle2modes} (resp. \cref{existuniqtriangle3modes}) with the conditions on the angles satisfied. Then the uniqueness for $n=2$ and $n=3$ follows from \cref{soundnessrewriting} and the uniqueness given by \cref{existuniqtriangle2modes} and \cref{existuniqtriangle3modes} respectively.
  For $n=1$, the proof is straightforward given \cref{soundnessrewriting}.
\end{proof}

\begin{figure}[tb]
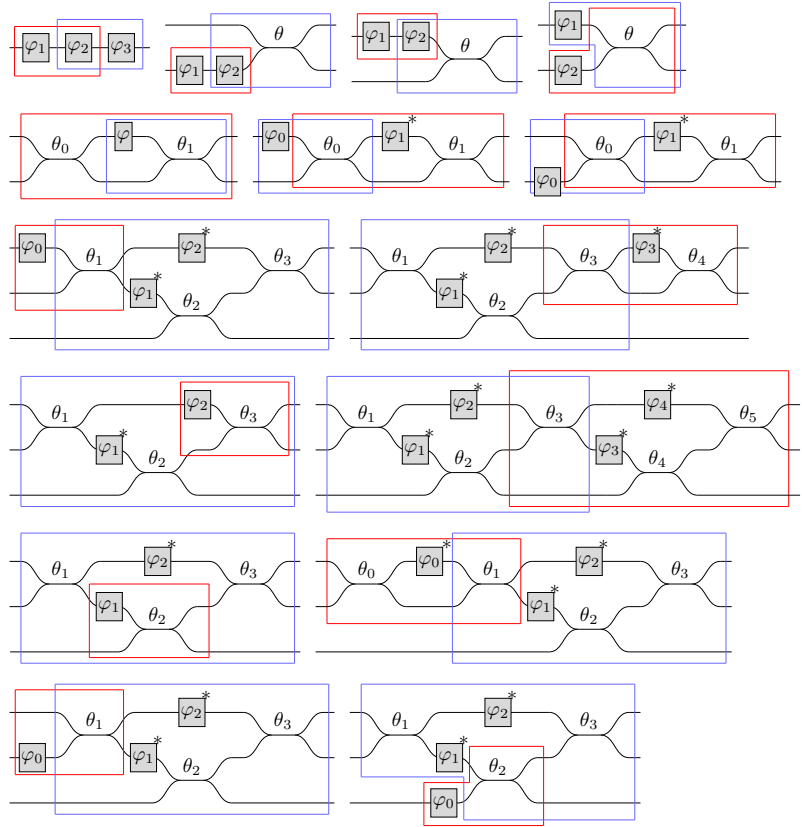

  \centering
  \scalebox{.75}{
    \begin{minipage}{\textwidth}
      \input{CP-1.tikz}\input{CP-2.tikz}\input{CP-3.tikz}\input{CP-12.tikz}\\[0.3cm]\input{CP-10.tikz}\input{CP-11.tikz}\input{CP-13.tikz}\\[0.3cm]\input{CP-4.tikz}\input{CP-9.tikz}\\[0.3cm]\input{CP-6.tikz}\input{CP-8.tikz}\\[0.3cm]\input{CP-5.tikz}\input{CP-7.tikz}\\[0.3cm]\input{CP-14.tikz}\input{CP-15.tikz}
    \end{minipage}}
  \caption{Non-trivial critical peaks. All angles are in $(0,2\pi)$ except for the
    two identities and the rules (\ref{phasemod2pi}) and
    (\ref{bsmod2pi}); additional
    constraints on the angles may occur for some rewrite rules to
    apply.} 
  \label{fig:peaks}
\end{figure}

\begin{lemma}\label{lem:locconfluent}
  \textup{PPRS} is locally confluent.
\end{lemma}

\begin{proof}
First, note that the trivial critical pairs, in which the two rewrite rules are applied to disjoint patterns, can be closed in a straightforward way. 
Indeed, after doing any of the two transformations involved, the other one can be done independently, and the final result does not depend on which transformation was applied first.\newline

Additionally, the non-trivial critical pairs, shown in \cref{fig:peaks}, all involve at most three (spatial) modes.

Indeed, first, two overlapping patterns necessarily share at least one spatial mode, so that if they both involve at most two modes, then their union involves at most three modes. 
This implies that any non-trivial critical pair involving at least four modes must arise from at least one instance of Rule~\eqref{glissadeEulerscalaires}.

If the other rewrite step of the critical pair is not an instance of Rule~\eqref{glissadeEulerscalaires}, it would involve at most two modes.
For the union with the instance of Rule~\eqref{glissadeEulerscalaires} to involve four modes, it must involve exactly two modes and the two patterns must share only one mode. 
Consequently, their union is composed only of phase shifters. Since in the left-hand side of Rule~\eqref{glissadeEulerscalaires} there is at most one phase shifter on each mode, the union of the two patterns must be a single phase shifter.\footnote{Note that two patterns that overlap only by identities can be considered disjoint.} 
Moreover, for the two patterns to share not more than one mode, it must be on the top mode of one pattern and on the bottom mode of the other pattern. 
Namely, since the left-hand side of Rule~\eqref{glissadeEulerscalaires} does not have a phase shifter on the bottom mode, it must be on the top mode of the associated pattern and on the bottom mode of the other pattern. 
The only rule in which the left-hand side involves two modes and has a phase shifter on the bottom mode is Rule~\eqref{removebottomphase}, but this phase shifter cannot belong to a pattern corresponding to the left-hand side of Rule~\eqref{glissadeEulerscalaires} since it would be both on the top left and on the bottom left of a beam splitter at the same time, which is not possible.

Hence, any non-trivial critical pair involving at least four modes must arise from two instances of Rule~\eqref{glissadeEulerscalaires}. Since the left-hand side of this rule does not have a phase shifter on the bottom mode, the two patterns cannot share only one mode and must share at least two modes. Then since their union involves at least four modes, they share exactly two modes. These two modes are the top two of one pattern and the bottom two of the other pattern. There are two generators that act only on the bottom two modes in the left-hand side of Rule~\eqref{glissadeEulerscalaires} and therefore can be in the intersection of the two patterns: the phase shifter labeled with $\varphi_1$ if present, and the beam splitter labelled with $\theta_2$.  If the phase shifter labeled with $\varphi_1$ is in the intersection of the two patterns then it necessarily correspond to the phase shifter labeled with $\varphi_2$ in the other pattern, but this is not possible since on is on the top right of a beam splitter whereas the other is on the bottom right of a beam splitter. Therefore, the two patterns necessarily overlap by one beam splitter, which is the bottom one in one pattern and one of the two top ones in the other pattern. But in the left-hand side of Rule~\eqref{glissadeEulerscalaires}, the bottom beam splitter is connected by its top wires to the bottom of another beam splitter on each side, whereas each of the two top ones is connected by at least one of its top wires to the top wire of another beam splitter, hence the two patterns cannot overlap this way.

Thus, all non-trivial critical pairs involve at most three spatial modes. It follows from \cref{lem:123normalform} that any critical pair on at most three wires can be closed, which gives us the local confluence.
\end{proof}

\subsection{Irreducible Polarisation-Preserving \lov-Circuits Are PPRS Triangular Normal Forms}
\label{preuvelemnormalformarePPRS}

\newcommand{\bs}[2]{\operatorname{BS_{#1}^{#2}}}
\newcommand{\ps}[2]{\operatorname{PS_{#1}^{#2}}}
\newcommand{\bss}[2]{\operatorname{BS_{#1}^{*#2}}}
\newcommand{\pss}[2]{\operatorname{PS_{#1}^{*#2}}}
\newcommand{\psbss}[2]{\operatorname{PS-BS_{#1}^{*#2}}}

For rigour and clarity purposes, $\bss{}{i}$ (resp.\ $\pss{}{i}$) represents a beam splitter (resp.\ phase shifter) on the modes $(i,i+1)$ (resp.\ on the mode $i$) with a non-zero angle (resp.\ a non-zero phase), or the identity otherwise. 
We denote as $\psbss{}{i}$ a  $\bss{}{i}$ with a $\pss{}{i}$ on its top left which is also the identity when the angle of the beam splitter is in $\{0,\frac{\pi}{2}\}$.
$C^{(i,j)}$ represents a circuit $C$ whose every component is between the modes $i$ and $j$.

We'll prove by induction that for any $n \in \mathbb{N}^*$ any irreducible polarisation-preserving \lov-circuit of size $n$ is a PPRS triangular normal form, given in \cref{def:pprs-nf}.
 
First, we can show that an irreducible circuit satisfies all the following properties:
\begin{itemize}
  \item No phase or angle can be outside of $(0,2\pi)$. This follows from Rules \eqref{phasemod2pi},\eqref{bsmod2pi},\eqref{zerophaseshifts} and \eqref{zerobs}.
  \item There are no consecutive phase shifters. This follows from Rule \eqref{fusionphaseshifts}.
  \item Angles for beam splitters and phases on the top left of beam splitters are in $(0,\pi)$. This follows from Rules \eqref{passagephasepi} and \eqref{soustractionpi}.
  \item There is no phase on the top left of a $\frac{\pi}{2}$ beam splitter. This follows from Rule \eqref{passagepisur2}.
  \item There is no phase on the bottom left of a beam splitter. This follows from Rule \eqref{removebottomphase}.
  \item There aren't two consecutive beam splitters on the same modes. This follows from Rule \eqref{fusionEulerbsphasebs} and the fact that there is no phase on the bottom left of a beam splitter.
\end{itemize}

Second, we can show that an irreducible circuit of size $n\geq 3$ is a PPRS triangular normal form if:  

\begin{enumerate}[label=\textbf{(H)}]
  \item \label{Hypothesis} The circuit can be decomposed as $P^{(1,n-1)} \circ D$ where: \begin{enumerate}[label=\footnotesize{\textbf{(H\arabic*)}}] 
    \item \label{Hypothesis1}$D=\pss{}{0} \circ \psbss{}{0} \circ \psbss{}{1} \circ ... \circ \psbss{}{n-2}$   
    \\ $D$ satisfies $\left( bs_D(0)\leq 1 \right) \wedge \left(\forall k \in \llbracket 1;n-2 \rrbracket :  bs_D(k)\leq bs_D(k-1) \right)$
    and has the shape of an upper-left anti-diagonal of $\psbss{}{}$ with one $\pss{}{0}$ at the end.
    \item \label{Hypothesis2} $P^{(1,n-1)}$ is an irreducible circuit on the $n-1$ other lower modes satisfying \ref{Hypothesis}. That implies that $P^{(1,n-1)}$ is a PPRS triangular normal form.
  \end{enumerate}\end{enumerate}

We'll therefore prove by induction that any irreducible circuit satisfies for $n\geq 3$ satisfies \ref{Hypothesis}, and therefore, any irreducible circuit of size $n\geq 3$ is a PPRS triangular normal form.
The case for $n\in\{1,2\}$ follows directly from \cref{lem:123normalform}.

\begin{proof}

The case for $n=3$ is directly induced from \cref{lem:123normalform}.

Let's consider an irreducible polarisation-preserving \lov-circuit $C$ of size $n+1$.  

We can show that there is necessarily at most one upper beam splitter on the first two modes, i.e.\ $C$ can't contain the pattern $BS_1^{0} \circ \pss{}{0} \circ P^{(1,n)} \circ  BS_2^{0}$. 
If such a pattern exists in an irreducible circuit, with our induction hypothesis $P^{(1,n)}$ is a PPRS triangular normal form. If there is an upper beam splitter $BS_P^{1}$ in $P^{(1,n)}$, then we could use Rule~\eqref{glissadeEulerscalaires} with $BS_1^{0}$ and $BS_2^{0}$. If the upper is the identity, then $BS_1^{0}$ and $BS_1^{0}$ would be consecutive beam splitters.
Therefore that pattern can't exist, and we necessarily have one beam splitter at most on the first mode.

If there is no upper beam splitter, $bs(0)=0$ and the first mode is at most one phase shifter.
By the induction hypothesis, the $n$ other wires form a PPRS triangular normal form $P^{(1,n)}$ of size $n$. As $bs(0)=0$, $C = P^{(1,n)} \circ \pss{}{0}$. Therefore \ref{Hypothesis} is satisfied.

If there is an upper beam splitter, then $C = P_2^{(1,n)} \circ \pss{}{0} \circ \psbss{}{0} \circ P_1^{(1,n)}$ where $bs(0)=1$ and $P_{i\in\{1,2\}}^{(1,n)}$ is a PPRS triangular normal form by the induction hypothesis.
Therefore $P_1^{(1,n)}= P_3^{(2,n)} \circ D_1 $.  $P^{(1,n)} = P_2^{(1,n)} \circ P_3^{(2,n)}$ is an irreducible circuit of size $n$ and therefore satisfies \ref{Hypothesis}. By taking $D^{(0,n)}= \pss{}{0} \circ  \psbss{}{0} \circ D_1 $, we can check that $C= P^{(1,n)} \circ D^{(0,n)}$ and that $C$ satisfies \ref{Hypothesis}.

Therefore, for any irreducible circuit of size $n\geq 3$, \ref{Hypothesis} is satisfied.

\end{proof}

Thus, with \cref{lem:123normalform}, for any $n\in \mathbb{N}^*$, PPRS triangular normal forms of size $n$ are the only irreducible polarisation-preserving \lov-circuits of size $n$.

\subsection{Completeness of the Polarisation-Preserving Fragment}
\label{preuvethmPPRScompleteness}

Let $N:n\rightarrow n$ be a PPRS triangular normal form of size $n$, and $I_n$ the identity circuit with $n$ identity wires.
We'll prove by induction that for any $n\in \mathbb{N}^*$: $\interps{N}=\interps{I_n} \Rightarrow N=I_n$.
For $n=1$, $N$ is \cref{PPRS-NF1}. The phase is necessarily zero. Therefore $N=I_1$.

Let's consider the case for a circuit of size $n+1$.
As we can see in the \cref{triangleNF}, there is only at most one beam splitter interacting with the upper mode. Therefore, we know the output of the first spatial mode:
$\ket{c_0} \mapsto e^{i\beta_{0,0}+\gamma_0}\cos (\alpha_{0,0})\ket{c_0} $.
To have the identity operation, we necessarily have $\alpha_{0,0}$.
By definition of the normal form, $\beta_{0,0}=0$. Thus, $\gamma_0=0$, and the first wire is the identity, whereas the other $n$ wires form a PPRS triangular normal form of size $n$.
The induction hypothesis implies that the $n$ other wires are necessarily the identity, which concludes the proof of the induction step.

Therefore, for any $n\in \mathbb{N}^*$: $\interps{N}=\interps{I_n} \Rightarrow N=I_n$.

\subsection{Soundness of the Rewrite System with Respect to the Equational Theory: Proof of Lemma \ref{soundnessrewriting}}
\label{preuvesoundnessrewriting}

  It suffices to show that for each rule of \cref{rulestriangleform}, we can transform the left-hand side into the right-hand side using the axioms of the $\lov$-calculus.
  
  The soundness of 
Rules~\eqref{phasemod2pi} and \eqref{bsmod2pi}
is a direct consequence of 
  \cref{2piperiodic}. Note that in both cases, transforming the left-hand side into the right-hand side using the equations of \cref{axiomsLOphotpres} only requires a bounded number of rewrite steps (see \cref{moduloborne} in \cref{preuve2piperiodic}).

  The soundness of Rule \eqref{fusionphaseshifts} is a direct consequence of \cref{phaseaddition}.
  
  The soundness of Rule \eqref{zerophaseshifts} is a direct consequence of \cref{phase0}.
  
  The soundness of Rule \eqref{zerobs} is a direct consequence of \cref{bs0}.
  
  The soundness of Rule \eqref{removebottomphase} is a direct consequence of Equations \eqref{phase0}, \eqref{phaseaddition} and \eqref{globalphasepropagationbs}.

  The soundness of Rule \eqref{passagepisur2} is a direct consequence of \cref{bspisur2swap}, \eqref{phaseaddition} and \eqref{phase0}.
  
  To prove the soundness of Rule \eqref{passagephasepi}, if $\varphi\in[\pi,2\pi)$ and $\theta\in(0,\pi)$, then we have:
  \begin{longtable}{RCL}
  \input{convtp-phasehbs.tikz}&\eqeqref{bs0}&\input{convtp-bs0phasehbs-ind0.tikz}\\\\
  &\eqeqref{Eulerbsphasebs}&\input{convtp-thetamoinspihbspimoinsphiphase0hpib-ind0.tikz}\\\\
  &\eqeqref{phase0}&\input{convtp-thetamoinspihbspimoinsphipib.tikz}
  \end{longtable}
  
  To prove the soundness of Rule \eqref{soustractionpi}, if $\theta\in[\pi,2\pi)$ then we have:
  \begin{longtable}{RCL}
  \input{bs.tikz}&\eqdeuxeqref{bs0}{phase0}&\input{convtp-bsphase0hbs0.tikz}\\\\
  &\eqeqref{Eulerbsphasebs}&\input{convtp-phase0bsphimoinspipis.tikz}\\\\
  &\eqeqref{phase0}&\input{convtp-bsphimoinspipis.tikz}
  \end{longtable}
  
  The soundness of Rule \eqref{glissadeEulerscalaires} is a direct consequence of Equations \eqref{Eulerscalaires} and \eqref{phase0}.

  The soundness of Rule \eqref{fusionEulerbsphasebs} is a direct consequence of Equations \eqref{Eulerbsphasebs} and \eqref{phase0}.

\subsection{Proof of Lemma \ref{putinNF}}
\label{putinNFproof}
By \cref{thm:PPRSuniquenormalform} and \cref{soundnessrewriting}, it suffices to prove that any circuit $D:n\to n$ without $\begin{tikzpicture}
	\begin{pgfonlayer}{nodelayer}
		\node [style=none] (2) at (1.1, 0) {};
		\node [style=gene] (3) at (0, 0) {$0$};
	\end{pgfonlayer}
	\begin{pgfonlayer}{edgelayer}
		\draw (3) to (2.center);
	\end{pgfonlayer}
\end{tikzpicture}
$ or $\begin{tikzpicture}
	\begin{pgfonlayer}{nodelayer}
		\node [style=none] (2) at (-1, 0) {};
		\node [style=detector] (3) at (0, 0) {$0$};
	\end{pgfonlayer}
	\begin{pgfonlayer}{edgelayer}
		\draw (3) to (2.center);
	\end{pgfonlayer}
\end{tikzpicture}
$ can be put in the form
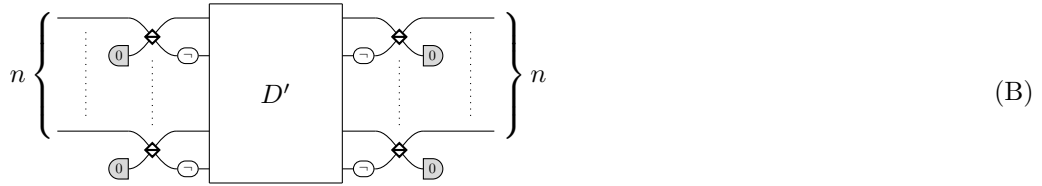
\begin{eqnABC}\label{shapepureNFDprime}\input{NFblocsDprimeexpl.tikz}\end{eqnABC}
where $D'$ is a polarisation-preserving \lov-circuit, by using the equations of \cref{axiomsLOphotpres}.

  Note that any circuit $D:n\to n$ without $\begin{tikzpicture}
	\begin{pgfonlayer}{nodelayer}
		\node [style=none] (2) at (1.1, 0) {};
		\node [style=gene] (3) at (0, 0) {$0$};
	\end{pgfonlayer}
	\begin{pgfonlayer}{edgelayer}
		\draw (3) to (2.center);
	\end{pgfonlayer}
\end{tikzpicture}
$ or $\begin{tikzpicture}
	\begin{pgfonlayer}{nodelayer}
		\node [style=none] (2) at (-1, 0) {};
		\node [style=detector] (3) at (0, 0) {$0$};
	\end{pgfonlayer}
	\begin{pgfonlayer}{edgelayer}
		\draw (3) to (2.center);
	\end{pgfonlayer}
\end{tikzpicture}
$ can be written as $d_k\circ\cdots\circ d_1$, with the $d_i$ of the form $I_\ell\oplus g\oplus I_{\ell'}$, where $I_\ell\coloneqq\input{Ilaccoladeg.tikz}$ (with $I_0=\begin{tikzpicture}[scale=0.3]
	\begin{pgfonlayer}{nodelayer}
		\node [style=sdiagrammevide] (0) at (0, 0) {};
	\end{pgfonlayer}
\end{tikzpicture}
$), $g\in\{\input{bs-xs.tikz},\input{beamsplitter-xs.tikz},\begin{tikzpicture}[scale=0.3]
	\begin{pgfonlayer}{nodelayer}
		\node [style=none] (5) at (-1.6, 0) {};
		\node [style=none] (9) at (1.6, 0) {};
		\node [style=PolRot,scale=0.75] (11) at (0, 0) {\small $\varphi$};
	\end{pgfonlayer}
	\begin{pgfonlayer}{edgelayer}
		\draw (5.center) to (11);
		\draw (11) to (9.center);
	\end{pgfonlayer}
\end{tikzpicture}
,\begin{tikzpicture}[scale=0.3]
	\begin{pgfonlayer}{nodelayer}
		\node [style=none] (5) at (-1.4, 0) {};
		\node [style=none] (9) at (1.4, 0) {};
		\node [style=PhS,scale=0.85] (11) at (0, 0) {};
		\node [style=none] (12) at (1, -1) {\scriptsize $\theta$};
	\end{pgfonlayer}
	\begin{pgfonlayer}{edgelayer}
		\draw (5.center) to (11);
		\draw (11) to (9.center);
	\end{pgfonlayer}
\end{tikzpicture}
,\input{swap-xs.tikz}\}$ 
  and $\ell+\ell'=n-1\text{ or }n-2$ depending on the type of $g$ (if $k=0$ then we take the product $d_k\circ\cdots\circ d_1$ to be the identity circuit $I_n$).
  
By Equations \eqref{negneg}, \eqref{pbspbs} and \eqref{halterevide}, $I_n$ is equivalent to the circuit of the form \eqref{shapepureNFDprime} with $D'=I_{2n}$, which is indeed polarisation-preserving.  It remains to prove that for any circuit $D$ of the form \eqref{shapepureNFDprime} with $D'$ polarisation-preserving, any $g\in\{\input{bs-xs.tikz},\input{beamsplitter-xs.tikz},\begin{tikzpicture}[scale=0.3]
	\begin{pgfonlayer}{nodelayer}
		\node [style=none] (5) at (-1.6, 0) {};
		\node [style=none] (9) at (1.6, 0) {};
		\node [style=PolRot,scale=0.75] (11) at (0, 0) {\small $\varphi$};
	\end{pgfonlayer}
	\begin{pgfonlayer}{edgelayer}
		\draw (5.center) to (11);
		\draw (11) to (9.center);
	\end{pgfonlayer}
\end{tikzpicture}
,\begin{tikzpicture}[scale=0.3]
	\begin{pgfonlayer}{nodelayer}
		\node [style=none] (5) at (-1.4, 0) {};
		\node [style=none] (9) at (1.4, 0) {};
		\node [style=PhS,scale=0.85] (11) at (0, 0) {};
		\node [style=none] (12) at (1, -1) {\scriptsize $\theta$};
	\end{pgfonlayer}
	\begin{pgfonlayer}{edgelayer}
		\draw (5.center) to (11);
		\draw (11) to (9.center);
	\end{pgfonlayer}
\end{tikzpicture}
,\input{swap-xs.tikz}\}$ and any $\ell$, the circuit $D\circ(I_\ell\oplus g\oplus I_{\ell'})$ can be put again in the form \eqref{shapepureNFDprime} with $D'$ polarisation-preserving.

  \def\btikzfig#1{\scalebox{0.8}{\input{#1.tikz}}}

The generator $g$ passes through the left part of $D$ as follows:
  \begin{equation}\label{passagephaseshift}\btikzfig{convtp-phssplitpol}\ =\ \btikzfig{convtp-splitpolphs}\end{equation}
  
  \begin{equation}\label{passagepolrot}\btikzfig{polrotsplitpol}\ =\ \btikzfig{splitpolbstheta}\end{equation}
  
  \begin{equation}\label{passagepbs}\btikzfig{pbssplitpol}\ =\ \btikzfig{splitpol3swaps}\end{equation}
  
  \begin{equation}\label{passagebs}\btikzfig{convtp-bssplitpol}\ =\ \btikzfig{convtp-splitpolswapsbs}\end{equation}
  
  \begin{equation}\label{passageswap}\btikzfig{swapsplitpol}\ =\ \btikzfig{splitpol4swaps}\end{equation}
  
Then, since by Equations \eqref{bspisur2swap}, \eqref{phaseaddition} and \eqref{phase0}, $\btikzfig{swap}=\btikzfig{bspissur2}$, we can 
remove the swaps in order to turn the middle part into a
polarisation-preserving circuit 
which finishes the proof.\bigskip
 
  It remains to prove \cref{passagephaseshift,passagepolrot,passagepbs,passagebs,passageswap} using the axioms of the $\lov$-calculus.
 
  To prove \cref{passagephaseshift}, we have:
  \begin{longtable}{RCL}
  \btikzfig{convtp-phssplitpol}&\eqeqref{absorptionphaseshiftleft}&\btikzfig{convtp-phsdoublesplitpol}\\\\
  &\eqdeuxeqref{gphpropapbs}{commutationnegphaseshift}&\btikzfig{convtp-splitpolphs}
  \end{longtable}
  
  To prove \cref{passagepolrot}, we have:
  \begin{longtable}{RCL}
  \btikzfig{polrotsplitpol}&\eqeqref{absorptionpolrotleft}&\btikzfig{polrotdoublesplitpol}\\\\
  &\eqeqref{polrotsfrombs}&\btikzfig{splitpolbsthetanbpbspbsnb}\\\\
  &\eqdeuxeqref{pbspbs}{negneg}&\btikzfig{splitpolbstheta}
  \end{longtable}
  
  To prove \cref{passagepbs}, we have:
  \begin{longtable}{RCL}
  \btikzfig{pbssplitpol-s}&=&\btikzfig{passagepbs1}\\\\
  &\eqeqref{absorptionpbsleft}&\btikzfig{passagepbs2}\\\\
  &\eqeqref{bigebrepbs}&\btikzfig{passagepbs3}\\\\
  &=&\btikzfig{passagepbs4}\\\\
  &\eqeqref{pbsswap}&\btikzfig{splitpol3swaps-s}
  \end{longtable}
  
  To prove \cref{passagebs}, we have:
  \begin{longtable}{RCL}
  \btikzfig{convtp-bssplitpol-s}&=&\btikzfig{convtp-passagebs1}\\\\
  &\eqeqref{absorptionbsleft}&\btikzfig{convtp-passagebs2}\\\\
  &\eqeqref{echangebspbs}&\btikzfig{convtp-passagebs3}\\\\
  &\eqeqref{commutationnegsbs}&\btikzfig{convtp-passagebs4}\\\\
  &=&\btikzfig{convtp-splitpolswapsbs-s}
  \end{longtable}
  
  Equation \eqref{passageswap} is by mere deformation.

\subsection{Equality of Unitary Transformations on a Subspace}
\label{egaliteaunitairespresproof}

In this section we show that if two unitary maps coincide on some subspaces then they are equal up to unitaries on the orthogonal subspaces:

\begin{lemma}\label{egaliteaunitairespres}
  Let $\Hi$ be a Hilbert space, $U,U'\colon\Hi\to\Hi$ be two unitary
  maps, and let $\Hi=\Hi_0^{\mathrm{in}}\oplus\Hi_1^{\mathrm{in}}$ and
  $\Hi=\Hi_0^{\mathrm{out}}\oplus\Hi_1^{\mathrm{out}}$ be two
  decompositions of $\Hi$ into orthogonal subspaces. Given any
  subspace $\Hi'$ of $\Hi$, we denote by $\pi_{\Hi'}\colon\Hi\to\Hi'$
  the orthogonal projector on $\Hi'$ and by
  $\iota_{\Hi'}\colon\Hi'\to\Hi$ the canonical injection. If
  $\pi_{\Hi_0^{\mathrm{out}}}\circ
  U\circ\iota_{\Hi_0^{\mathrm{in}}}=\pi_{\Hi_0^{\mathrm{out}}}\circ
  U'\circ\iota_{\Hi_0^{\mathrm{in}}}$, then there exists two unitary
  maps
  $Q_{\mathrm{in}}\colon\Hi_1^{\mathrm{in}}\to\Hi_1^{\mathrm{in}}$ and
  $Q_{\mathrm{out}}\colon\Hi_1^{\mathrm{out}}\to\Hi_1^{\mathrm{out}}$
  such that
  $U'=(I\oplus Q_{\mathrm{out}})\circ U\circ (I\oplus
  Q_{\mathrm{in}})$.
\end{lemma}

\begin{proof}

  We denote $U_0\coloneqq U\circ\iota_{\Hi_0^{\mathrm{in}}}$, $U_{00}\coloneqq\pi_{\Hi_0^{\mathrm{out}}}\circ U\circ\iota_{\Hi_0^{\mathrm{in}}}$ 
  and $U_{01}\coloneqq\pi_{\Hi_1^{\mathrm{out}}}\circ U\circ\iota_{\Hi_0^{\mathrm{in}}}$. We also define analogous notations for $U'$. Note that $U_0$ and $U'_0$ are isometries. For any $v,v'\in\Hi_0^{\mathrm{in}}$, one has $\scalprod{v}{v'}=\scalprod{U_{0}(v)}{U_{0}(v')}=\scalprod{U_{00}(v)}{U_{00}(v')}+\scalprod{U_{01}(v)}{U_{01}(v')}$. Similarly, $\scalprod{v}{v'}=\scalprod{U'_{00}(v)}{U'_{00}(v')}+\scalprod{U'_{01}(v)}{U'_{01}(v')}$. Since $U_{00}=U'_{00}$, this implies that \begin{equation}\label{memeprodscals}\forall v,v'\in\Hi_0^{\mathrm{in}},\quad\scalprod{U_{01}(v)}{U_{01}(v')}=\scalprod{U'_{01}(v)}{U'_{01}(v')}.\end{equation} Let $v_1,...,v_d\in\Hi_0^{\mathrm{in}}$ such that $U_{01}(v_1),...,U_{01}(v_d)$ is an orthonormal basis of the image $U_{01}(\Hi_0^{\mathrm{in}})$ of $U_{01}$. By \eqref{memeprodscals}, $U'_{01}(v_1),...,U'_{01}(v_d)$ is an orthonormal basis of $U'_{01}(\Hi_0^{\mathrm{in}})$. Let $Q_{\mathrm{out}}\colon\Hi_1^{\mathrm{out}}\to\Hi_1^{\mathrm{out}}$ be any unitary map such that $\forall i\in\{1,...,d\}, Q_{\mathrm{out}}(U_{01}(v_i))=U'_{01}(v_1)$. For any $v\in\Hi_0^{\mathrm{in}}$, there exist $\lambda_1,...,\lambda_d\in\CC$ such that $U_{01}(v)=\sum_{i=1}^d\lambda_iU_{01}(v_i)$. Then by \eqref{memeprodscals}, $\|U'_{01}(v)-\sum_{i=1}^d\lambda_iU'_{01}(v_i)\|=\|U_{01}(v)-\sum_{i=1}^d\lambda_iU_{01}(v_i)\|=0$, so that $U'_{01}(v)=\sum_{i=1}^d\lambda_iU'_{01}(v_i)$. Hence, $Q_{\mathrm{out}}(U_{01}(v))=U'_{01}(v)$. Thus, $U'_{01}=Q_{\mathrm{out}}\circ U_{01}$. Since $U_0=U_{00}+U_{01}$ and $U'_0=U'_{00}+U'_{01}$, this implies that $U'_0=(I\oplus Q_{\mathrm{out}})\circ U_0$.
  
  In other words $\forall v\in\Hi_0^{\mathrm{in}}, U'(v)=(I\oplus Q_{\mathrm{out}})\circ U(v)$. Hence, $U'(\Hi_0^{\mathrm{in}})=(I\oplus Q_{\mathrm{out}})\circ U(\Hi_0^{\mathrm{in}})$, so that since $\Hi_0^{\mathrm{in}}$ and $\Hi_1^{\mathrm{in}}$ are the orthogonal complement of each other and $U,U'$ are unitary, we also have $U'(\Hi_1^{\mathrm{in}})=(I\oplus Q_{\mathrm{out}})\circ U(\Hi_1^{\mathrm{in}})$ (which is the orthogonal complement of $U'(\Hi_0^{\mathrm{in}})$). Let $w_1,...,w_k$ be an orthonormal basis of $\Hi_1^{\mathrm{in}}$, and for every $i\in\{1,...,k\}$, let $w'_i\coloneqq {U'}^\dag\circ(I\oplus Q_{\mathrm{out}})\circ U(w_i)$. The fact that $U'(\Hi_1^{\mathrm{in}})=(I\oplus Q_{\mathrm{out}})\circ U(\Hi_1^{\mathrm{in}})$ implies that $w'_1,...,w'_k$ is also an orthonormal basis of $\Hi_1^{\mathrm{in}}$. Let $Q_{\mathrm{in}}\colon\Hi_1^{\mathrm{in}}\to\Hi_1^{\mathrm{in}}$ be the unique unitary map such that for all $i\in\{1,...,k\}$, $Q_{\mathrm{in}}(w'_i)=w_i$. Then $U'=(I\oplus Q_{\mathrm{out}})\circ U\circ (I\oplus Q_{\mathrm{in}})$, as desired.
 \end{proof}

\subsection{Universality of \lov-Circuits}

\label{appendix:universality}

Let $U:\mathbb C^{M_n}\to \mathbb C^{M_m}$ be a sub-unitary map i.e. a map $U$ s.t. $U^\dagger U\sqsubseteq I_n$. We show in the following how to construct a \lov-circuit $C$ s.t. $\interp C = U$. 
First notice that $V: \mathbb C^{2n}\to \mathbb C^{2m} =  \mu_m\circ U\circ \mu_n^{\dagger}$  is also a sub-unitary map, where $\mu_n: C^{M_{n}}\to \mathbb C^{2n}$ is such that $\mu \ket{\V_k}= \ket {2k}$ and $\mu \ket{\H_k}= \ket {2k+1}$.

Since $I_n-V^\dagger V$ is semi-definite positive there exists $A: \mathbb C^{2n}\to \mathbb C^{\ell}$ s.t. $A^\dagger A =  I_n-V^\dagger V$. As a consequence the matrix $W: \mathbb C^{2n}\to \mathbb C^{2m+\ell} =\left(\begin{array}{c}V\\A\end{array}\right)$ is an isometry since $W^\dagger W = V^\dagger V+A^\dagger A = I_{2n}$. $W$ can be turned into a unitary matrix by adding columns to $W$, i.e. $\exists B,D$ s.t. $U':\mathbb C^{2m+\ell}\to \mathbb C^{{2m+\ell}}  = \left(\begin{array}{cc}V&A\\B&D\end{array}\right)$ is unitary. By \cref{triangleuniversel}, there exists a \lopp-circuit $T$ s.t. $\interph{T} = U'$. Let $C$ be the following \lov-circuit: 
\[\input{NFblocs1majexpl-univ.tikz}\]
By construction, $\interp C = U$. \qed
\section{Other Valid Equations}

\begin{figure}[hbp]
\centering
\scalebox{.8333}{\begin{minipage}{1.2\textwidth}
  \begin{multicols}{2}
  \begin{equation}\label{polrot0}\begin{array}{rcl}\input{pol-rot0.tikz}&=&\input{filcourt.tikz}\end{array}\end{equation}
  \vspace{\nspazer}
  \begin{equation}\label{phasepi}\begin{array}{rcl}\input{pol-rotpi.tikz}&=&\input{phase-shiftpi.tikz}\end{array}\end{equation}
  \vspace{\nspazer}
  \begin{equation}\label{bspi}\begin{array}{rcl}\input{bspi.tikz}&=&\input{pisurpi.tikz}\end{array}\end{equation}
  \vspace{\nspazer}
  \begin{equation}\label{pibsmoinsphi}\begin{array}{rcl}\input{convtp-pibbs.tikz}&=&\input{convtp-bspibmoinsphi.tikz}\end{array}\end{equation}
  \vspace{\nspazer}
  \begin{equation}\label{commutationphasepolrot}\begin{array}{rcl}\input{convtp-phasetheta1polrottheta2.tikz}&=&\input{convtp-polrottheta2phasetheta1.tikz}\end{array}\end{equation}
  \vspace{\nspazer}
  \begin{equation}\label{additionpolrot}\begin{array}{rcl}\input{polrottheta1polrottheta2.tikz}&=&\input{polrottheta1plustheta2.tikz}\end{array}\end{equation}
  \vspace{\nspazer}
  \begin{equation}\label{commutationpolrotbs}\begin{array}{rcl}\input{convtp-polrotsbs.tikz}&=&\input{convtp-bspolrots.tikz}\end{array}\end{equation}
  \vspace{\nspazer}
  \begin{equation}\label{commutationbspbs}\begin{array}{rcl}\input{convtp-bspbs.tikz}&=&\input{convtp-pbsbs.tikz}\end{array}\end{equation}
  \vspace{\nspazer}
  \end{multicols}
\begin{equation}\label{decompositionpbsavecHadamard}\begin{array}{rcl}\input{beamsplitter.tikz}&=&\input{pbsdecompHadamard.tikz}\end{array}\end{equation}
  \vspace{\nspazer}
\begin{equation}\label{phasecontrolsansancillas}\begin{array}{rcl}\input{schemaphasecontrolsansancillashaut.tikz}&=&\input{phasecontrolparpbs2thetasurfil.tikz}\end{array}\end{equation}
\end{minipage}}
  \caption{Interesting consequences of the axioms of the $\lov$-calculus.\label{interestingequations}}
  \end{figure}

\end{document}